\documentclass[a4paper,12pt,twoside]{report}
\usepackage[left=2cm,right=2cm,top=2cm,bottom=3cm]{geometry}

\usepackage{graphicx}
\usepackage{verbatim}
\usepackage{latexsym}
\usepackage{mathchars}
\usepackage{setspace}
\usepackage{amsmath}
\usepackage{amssymb}
\usepackage{bm}
\usepackage{hyperref}
\usepackage{verbatim}
\usepackage{alltt}
\usepackage{listings}
\usepackage{graphicx}
\usepackage{enumitem}
\usepackage{multicol}
\usepackage{mathtools}
\newcommand\myeq{\stackrel{\mathclap{\normalfont\mbox{def}}}{=}}
\usepackage{breqn}
\usepackage[utf8]{inputenc}
\usepackage[english]{babel}
\usepackage{amsthm}
\usepackage[usenames, dvipsnames]{color}
\usepackage[dvipsnames]{xcolor}

\newlist{notes}{enumerate}{1}
\setlist[notes]{label=Note: ,leftmargin=*}

\input{blocked.sty}
\input{uhead.sty}
\input{boxit.sty}
\input{icthesis.sty}

\newcommand{\NN}{{\sf I\kern-0.14emN}}   
\newcommand{\ZZ}{{\sf Z\kern-0.45emZ}}   
\newcommand{\QQQ}{{\sf C\kern-0.48emQ}}   
\newcommand{\RR}{{\sf I\kern-0.14emR}}   

\newcommand{\normallinespacing}{\renewcommand{\baselinestretch}{1.5} \normalsize}

\newcommand{\syncc}{~\stackrel{\textstyle \rhd\kern-0.57em\lhd}{\scriptstyle L}~}

\newtheorem{definition}{Definition}[section]
\newtheorem{theorem}{Theorem}[section]
\newtheorem{conjecture}{Conjecture}[section]
\newtheorem{example}{Example}[section]
\newtheorem*{remark}{Remark}
\newtheorem*{remark1}{Remark 1}
\newtheorem*{remark2}{Remark 2}

\begin{document}

\title{\LARGE {\bf Weak Keys and Cryptanalysis of a Cold War Block Cipher}\\
 \vspace*{6mm}
}

\author{Marios Georgiou}
\submitdate{August 2018}

\normallinespacing
\maketitle

\preface
\addcontentsline{toc}{chapter}{Abstract}

\begin{abstract}

T-310 is a cipher that was used for encryption of governmental communications in East Germany during the final years of the Cold War. Due to its complexity and the encryption process, there was no published attack for a period of more than 40 years until 2018 by Nicolas T. Courtois et al. in [10]. In this thesis we study the so called 'long term keys' that were used in the cipher, in order to expose weaknesses which will assist the design of various attacks on T-310.

\end{abstract}
\cleardoublepage

\addcontentsline{toc}{chapter}{Acknowledgements}

\begin{acknowledgements}
I would like to thank my supervisor, Nicolas Tadeusz Courtois, for his guidance and support throughout the completion of this thesis. I would also like to thank Matteo Scarlata, a student at ETH Zurich, for the design of most of the figures that appear in this thesis.

\end{acknowledgements}

\body
\chapter{Introduction}

\section{Motivation and Goal}

T-310 was an extremely significant cipher as it was used for encryption of teletype communications during the Cold War in East Germany. Its production started in the 1970s and it was designed by mathematicians and engineers who specialised in cryptography and cryptology.\\
After it was examined for its security by ZCO and Soviet cryptologists in 1980, a large number of T-310 cipher machines started being massively produced, and the number of around 3800 active cipher machines being used in 1989 was reached. During the final years of the Cold War, T-310 had become so famous that it was used for encryption of governmental communicatios (see chapter 1 in [9]).\\\\
T-310 is a rather complex cipher and for more than 40 years no attack was published until 2018 by Nicolas T. Courtois et al. in [10]. The aim of the thesis is to study the so called 'long term keys' which were used for a long period of time - probably for a year [4] or when it was required [5] - in order to expose weaknesses of the cipher and exploit these weaknesses to produce invariances which will assist the design of future attacks on T-310.

\section{Structure of the Thesis}

In Chapter 2, we provide the reader with the necessary knowledge about T-310 and we explain how encryption with T-310 was performed. In the third Chapter, we study the Boolean function that was used inside the round function of T-310. In the next two Chapters we search for long term keys that are weak against Linear Cryptanalysis and Generalised Linear Cryptanalysis. These weak long term keys can then be used to deploy attacks on T-310. In the final Chapter, we summarise our results and we suggest how the research on the design of attacks on T-310 can be extended.
\chapter{Background Theory}

\label{ch:background}

In this chapter we will repeat what has already been stated in [9].
\vspace{-30pt}
\section{Introduction to T-310}
T-310 contains a block cipher, which is iterated many times $(13 \cdot 127 = 1651)$ in a stream encryption mode in order to extract only 10 bits, and use them to encrypt a single 5-bit character of the plaintext.\\ 
The size of the block cipher is 36 bits, the secret (short term) key has size 240 bits and the IV has a length of 61 bits. The secret key is halved into $s_{1-120,1}$ and $s_{1-120,2}$ and it is repeated every 120 rounds. In order to expand the IV bits, the designers of T-310 implemented an LFSR (Linear-Feedback Shift Register) which has a prime period of $2^{61}-1$ defined by:
$$f_i = f_{i-61} \oplus f_{i-60} \oplus f_{i-59} \oplus f_{i-56}$$

\begin{definition}[Long Term Key]
We call a long term key a triple $(D,P,\alpha)$, where $D:\{1,\dots,9\} \to \{0,\dots,36\}$, $P:\{1,\dots,27\} \to \{0,\dots,36\}$ and $\alpha \in \{1,\dots,36\}$
\end{definition}
There are different classes of long term keys, such as KT1 or KT2 (see Appendices A and B). The purpose of long term keys is to make T-310 operate in a different way each time a different long term key is used.

\begin{definition}[Round Function]
We denote by $u_{m,1-36}$ the block cipher's state on all 36 bits at a round $m = 0,1,2,\dots$ and we set the initial state $u_{0,1-36}=$ 0xC5A13E396.
We denote the round function by $\phi: \{0,1\}^3 \times \{0,1\}^{36} \to \{0,1\}^{36}$. Hence
$$u_{m,1-36} = \phi (s_{m,1},s_{m,2},f_m;u_{m-1,1-36})$$
\end{definition}
As we can see from Figure 2.1, T-310 is a Feistel cipher with 4 branches, and all bits that are numbered $1,\dots,36$ and are multiples of 4 are replaced, while those which are not multiples of 4 are just shifted to the next position. For example, bit 1 is shifted at position 2.
\begin{figure}[h]
    \centering
    \includegraphics[scale=0.6]{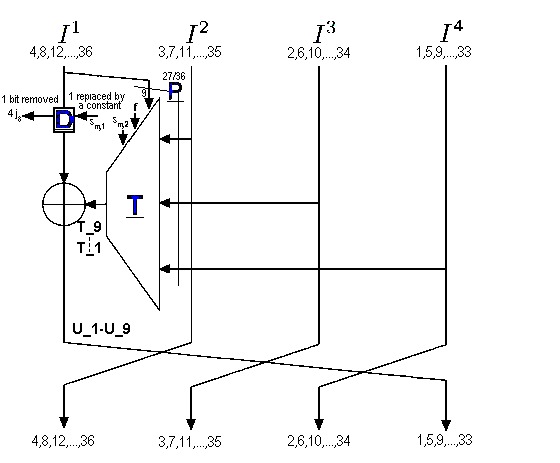}
    \caption{Internal Structure of an Encryption Round for T-310 for KT1 Keys}
    \label{fig:mesh1}
\end{figure}

\begin{definition}[Boolean Function Z]
T-310 uses the same Boolean function $Z:\mathbb{F}^6_2 \to \mathbb{F}_2$ four times in Z1,Z2,Z3 and Z4 (see Figure 2.2) and according to [15] is defined by:
\begin{align*}
Z(x0,x1,x2,x3,x4,x5) = x0 \oplus x4 \oplus x5 \oplus x0x3 \oplus x1x2 \oplus x1x4 \oplus x3x4 \oplus x4x5 \oplus x0x2x3\\ 
\oplus x0x2x5 \oplus x0x3x4 \oplus x1x2x5 \oplus x1x3x5 \oplus x2x4x5 \oplus x0x1x2x3\\
\oplus x0x1x2x4 \oplus x0x1x4x5 \oplus x1x2x3x5 \oplus x0x1x2x3x4 \oplus x0x2x3x4x5
\end{align*}
\end{definition}

We can now demonstrate how the bits $U_{1-9}$ are computed after one round. According to [9] whenever $D(i)=0$ we have $u_{m,0} \myeq s_{m+1,1}$. By the description of KT1 keys the bits $U_{1-9}$ are computed according to the following equations:

\setcounter{equation}{0}
\vskip-9pt
\vskip-9pt
\begin{align}
U_1 \oplus s_{1}&= U_2 \oplus u_{D(2)}  \oplus u_{P(27)}\tag{1}\\
U_2 \oplus u_{D(2)} &= U_3 \oplus u_{D(3)}  \oplus Z_4(u_{P(21-26)})\tag{2}\\
U_3 \oplus u_{D(3)} &= U_4 \oplus u_{D(4)}  \oplus u_{P(20)}\tag{3}\\
U_4 \oplus u_{D(4)} &= U_5 \oplus u_{D(5)}  \oplus Z_3(u_{P(14-19)}) \oplus s_{2}\tag{4}\\
U_5 \oplus u_{D(5)} &= U_6 \oplus u_{D(6)}  \oplus u_{P(13)}\tag{5}\\
U_6 \oplus u_{D(6)} &= U_7 \oplus u_{D(7)}  \oplus Z_2(u_{P(7-12)})\tag{6}\\
U_7 \oplus u_{D(7)} &= U_8 \oplus u_{D(8)}  \oplus u_{P(6)}\tag{7}\\
U_8 \oplus u_{D(8)} &= U_9 \oplus u_{D(9)}  \oplus Z_1(s_{2},u_{P(1-5)})\tag{8}\\
U_9 \oplus u_{D(9)} &= f~~~~~~~~~~~~ &\tag{9}
\end{align}

\begin{figure}[h]
    \centering
    \includegraphics[scale=0.27]{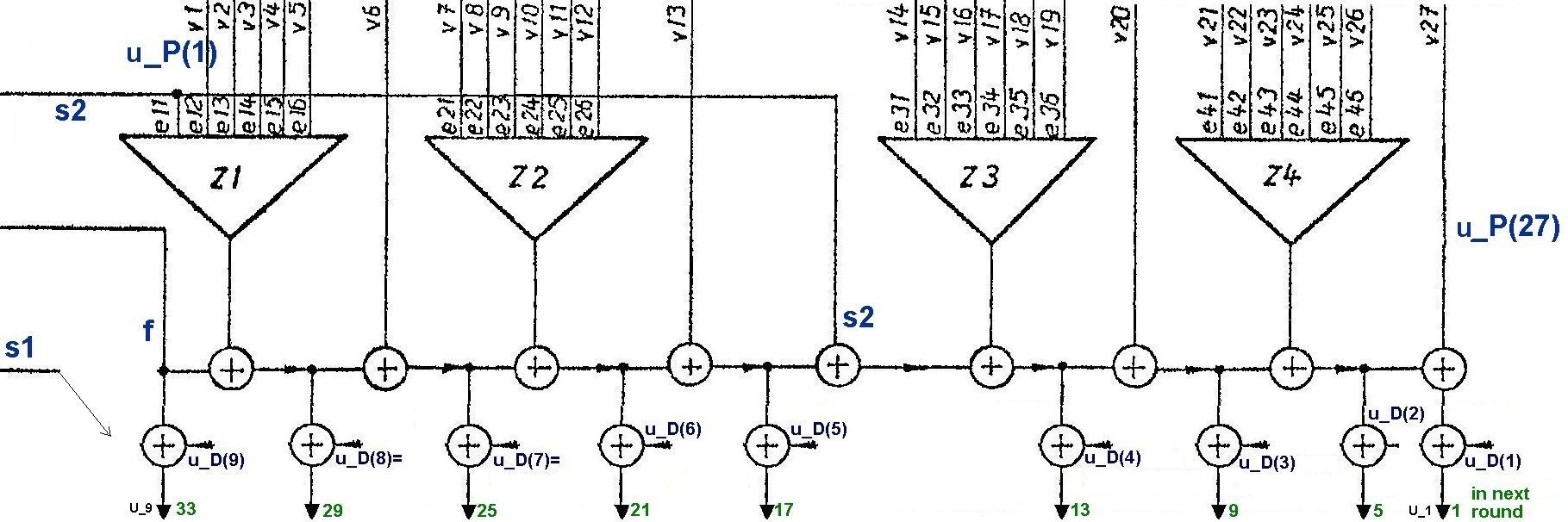}
    \caption{Internal Structure of One Round of T-310}
    \label{fig:mesh1}
\end{figure}

\pagebreak
\section{Encryption in T-310}

The constant $\alpha \in \{1, \dots ,36\}$ from the long term key defines which bit will be extracted after every 127 rounds in order to be used for the encryption. Thus
$$a_i \myeq u_{127\cdot i, \alpha}$$
After the collection of 13 bits, we dispose 3 of them and use the remaining as follows:
$$C_j = (P_j \oplus B_j) \cdot M^{r_j}$$
where $(P_j,C_j)$ is a 5-bit plaintext and ciphertext pair, $B_j = (a_{7+13(j-1)},\dots,a_{11+13(j-1)})$ are 5 consecutive bits out of the 13 we extracted and we use the remaining 5 consecutive bits to compute $r_j$ as follows:
$$r_l = \begin{cases}
0 & \quad \text{if}\,\,\,\,\,\,\,\,\,\,\,\,\,\,\,R_j = (0,0,0,0,0) \\
0 & \quad \text{if}\,\,\,\,\,\,\,\,\,\,\,\,\,\,\,R_j = (1,1,1,1,1) \\
31-r & \quad \text{if }R_j \cdot M^r = (1,1,1,1,1)
\end{cases}$$
where $R_j \myeq (a_{1+13(j-1)},\dots,a_{5+13(j-1)})$ and 
$$M = \begin{pmatrix}
0 & 0 & 0 & 0 & 1 \\
1 & 0 & 0 & 0 & 0 \\
0 & 1 & 0 & 0 & 1 \\
0 & 0 & 1 & 0 & 0 \\
0 & 0 & 0 & 1 & 0 
\end{pmatrix}, \text{   which is such that }M^{31}=Id$$

\vspace{50pt}
\section{Related Work}

The structure and the encryption process of T-310 was extensively studied in [9] by the acquisition of historical documents from 1970s. The long term KT1 keys and their properties were further studied in Appendix B of [9], while in Appendix C, Nicolas T. Courtois et al. studied the bijectivity of one round function $\phi$, and provided a mathematical proof that all KT1 keys induce a bijective round function. Furthermore, in [7] Nicolas T. Courtois presented a slide attack on T-310 in combination with a decryption oracle attack.\\\\
The notion of Linear Cryptanalysis, which was introduced in [6] by Matsui in order to design an attack for DES, has served as inspiration since then for the design of many attacks, including the attacks in section 4.3. Generalised Linear Cryptanalysis (GLC) was introduced by Carlo Harpes et al. in [2] and it was also studied in [8]. Generalised Linear Cryptanalysis can be used to break a Feistel cipher which is considered secure against Linear and Differential Cryptanalysis. Since T-310 is also a Feistel cipher with 4 branches, we were inspired to use GLC in order to find long term keys which are weak against it. Finally, we were also motivated by [12], where the authors present attacks on Toyocrypt and LILI-128, two stream ciphers that use LFSRs and non linear Boolean functions $f$, as T-310 also does, by exploiting the fact that there exist some multivariate polynomials $g$ such that $f\cdot g = 0$.
\chapter{Boolean Function of T-310}

It is expected that the Boolean function Z was designed in such a way that it is resistant against linear and differential cryptanalysis, so that an attacker will not be able to correlate some input bits with some output bits of the function. To check whether this is true, we need to calculate the Walsh Spectrum and the Autocorrelation Spectrum of the Boolean function. 

\begin{definition}[Walsh Spectrum]
The Walsh Spectrum of a Boolean function is defined to be the product of the Boolean function's truth table and the Hadamard Matrix $[14,16]$.
\end{definition}
\begin{definition}[Hadamard Matrix]
The Hadamard Matrix $(H_n)$ is defined to be an $n \times n$ square matrix, where $n=2^k$ for some $k \in \mathbb{N}$, whose entries are only $1$ or $-1$ and its rows are mutually orthogonal. Hadamard Matrices are defined by the following sequence:
\setlength{\belowdisplayskip}{5pt} \setlength{\belowdisplayshortskip}{5pt}
\setlength{\abovedisplayskip}{5pt}
\setlength{\abovedisplayshortskip}{5pt}
$$H_1 = [1]\,, \,\,\,\,\, H_2 = \begin{bmatrix}
1 & 1 \\
1 & -1 
\end{bmatrix}\,, 
\,\,\,\,\,
H_{2^k} = 
\begin{bmatrix}
H_{2^{k-1}} & H_{2^{k-1}} \\
H_{2^{k-1}} & -H_{2^{k-1}}
\end{bmatrix}
= H_2 \otimes H_{2^{k-1}}
$$
for $k \geq 2$, where $\otimes$ is the Kronecker product.
\end{definition}
\begin{definition}[Autocorrelation Spectrum]
The Autocorrelation Spectrum of a Boolean function $f$ is defined as 
\setlength{\belowdisplayskip}{5pt} \setlength{\belowdisplayshortskip}{5pt}
\setlength{\abovedisplayskip}{5pt}
\setlength{\abovedisplayshortskip}{5pt}
$$\hat{r}_f(x) = \sum_{y \in \mathbb{F}_2^n} (-1)^{f(y) \oplus f(y \oplus x)}$$
where $x \in \mathbb{F}_2^n$ $[1,3,13]$.
\end{definition}
\begin{table}[h]
\caption{Walsh Spectrum and Autocorrelation Spectrum for Boolean Function Z}
\begin{center}
{\small \begin{tabular}{|c|c|lllcc}
\cline{1-2} \cline{6-7}
Walsh Spectrum & Frequency &  &  & \multicolumn{1}{l|}{} & \multicolumn{1}{c|}{Autocorrelation Spectrum} & \multicolumn{1}{c|}{Frequency} \\ \cline{1-2} \cline{6-7} 
-10             & 1         &  &  & \multicolumn{1}{l|}{} & \multicolumn{1}{c|}{-24}                      & \multicolumn{1}{c|}{3}         \\ \cline{1-2} \cline{6-7} 
-8             & 2         &  &  & \multicolumn{1}{l|}{} & \multicolumn{1}{c|}{-16}                      & \multicolumn{1}{c|}{6}         \\ \cline{1-2} \cline{6-7} 
-6             & 8         &  &  & \multicolumn{1}{l|}{} & \multicolumn{1}{c|}{-8}                       & \multicolumn{1}{c|}{13}        \\ \cline{1-2} \cline{6-7} 
-4             & 5         &  &  & \multicolumn{1}{l|}{} & \multicolumn{1}{c|}{0}                        & \multicolumn{1}{c|}{20}        \\ \cline{1-2} \cline{6-7} 
-2              & 10        &  &  & \multicolumn{1}{l|}{} & \multicolumn{1}{c|}{8}                        & \multicolumn{1}{c|}{17}        \\ \cline{1-2} \cline{6-7} 
0              & 16        &  &  & \multicolumn{1}{l|}{} & \multicolumn{1}{c|}{16}                       & \multicolumn{1}{c|}{3}         \\ \cline{1-2} \cline{6-7} 
2              & 8         &  &  & \multicolumn{1}{l|}{} & \multicolumn{1}{c|}{24}                       & \multicolumn{1}{c|}{1}         \\ \cline{1-2} \cline{6-7} 
4              & 7         &  &  & \multicolumn{1}{l|}{} & \multicolumn{1}{c|}{64}                       & \multicolumn{1}{c|}{1}         \\ \cline{1-2} \cline{6-7} 
6              & 5         &  &  &                       &                                               &                                \\ \cline{1-2}
8             & 1         &  &  &                       &                                               &                                \\ \cline{1-2}
\end{tabular}}
\end{center}
\end{table}
As we can see in both tables, most values are focused around 0 which indicates that the Boolean function Z is resistant against linear and differential cryptanalysis.\\
By extending Theorem C.0.1 from [12], we have:
\begin{theorem}
Let $f$ be any Boolean function $f:GF(2)^k \to GF(2)$. There is a Boolean function $g \neq 0$ of degree at most $k$ such that $f(x) \cdot g(x) = (1+a)\cdot g(x)$, where $a \in \{0,1\}$.
\end{theorem}
\begin{proof}
The number of all possible inputs to a Boolean function of degree $k$ is $2^k$. For every Boolean function $f$ different from $f=0$ or $f=1$:
\setlength{\belowdisplayskip}{5pt} \setlength{\belowdisplayshortskip}{5pt}
\setlength{\abovedisplayskip}{5pt}
\setlength{\abovedisplayshortskip}{5pt}
$$\exists \text{ }a \in \{0,1\}\text{, such that } |\{x|f(x) = a\}|\text{ } \leq \text{ } \frac{1}{2}2^k$$
Let $X = \{x|f(x) = a\}$. We then construct a matrix whose rows are the elements of $X$ and the columns are all possible monomials of degree at most $k$. Each entry of the matrix is the value of the monomial with regard to the corresponding $x \in X$, which will result in the value $0$ or $1$. Clearly, the number of columns is:
$$\sum_{i=0}^{k} \binom ki = 2^k$$
We observe that the number of columns is greater than the number of rows, thus there should be a linear combination of monomials that is equal to $0$. Hence, for a Boolean function $g$ of degree of at most $k$:
\setlength{\belowdisplayskip}{0pt} \setlength{\belowdisplayshortskip}{0pt}
\setlength{\abovedisplayskip}{0pt}
\setlength{\abovedisplayshortskip}{0pt}
$$\exists \text{ }a \in \{0,1\}\text{, such that }\forall x \in X \text{, }f(x) = a 	\Rightarrow g(x) = 0$$
Therefore,
$$\forall x \in \{0,1\}^k \text{, } f(x) \cdot g(x) = (1 + a) \cdot g(x)$$
\end{proof}
\begin{remark1}
We can rewrite the equation in Theorem 3.0.1. as:
$$(f(x) + 1 + a) \cdot g(x) = 0$$
\end{remark1}
\begin{remark2}
There are 32 such $g$'s that satisfy the above equation for the Boolean function of T-310. We list one for each possible degree below.
\begin{itemize}
\item $g = x0*x1*x4 \oplus x0*x2*x3 \oplus x0*x2*x5 \oplus x0*x3*x4 \oplus x0*x3*x5 \oplus x0*x4*x5 \oplus x0 \oplus x1*x2*x3 \oplus x1*x2*x4 \oplus x1*x2*x5 \oplus x1*x2 \oplus x1*x3*x5 \oplus x1*x4*x5 \oplus x2*x3*x5 \oplus x2*x3 \oplus x4*x5 \oplus x4 \oplus x5 \oplus 1$
\item $g = x0*x1*x2*x3 \oplus x0*x1*x4 \oplus x0*x2*x3*x4 \oplus x0*x2*x3*x5 \oplus x0*x2*x4*x5 \oplus x0*x2*x5 \oplus x0*x3*x4*x5 \oplus x0*x3*x4 \oplus x0*x3*x5 \oplus x0 \oplus x1*x2*x3*x4 \oplus x1*x2*x4*x5 \oplus x1*x2*x4 \oplus x1*x2*x5 \oplus x1*x2 \oplus x4*x5 \oplus x4 \oplus x5 \oplus 1$
\item $g = x0*x1*x2*x3 \oplus x0*x1*x2*x4*x5 \oplus x0*x1*x4 \oplus x0*x2*x3*x5 \oplus x0*x2*x5 \oplus x0*x3*x4*x5 \oplus x0*x3*x4 \oplus x0*x3*x5 \oplus x0 \oplus x1*x2*x4*x5 \oplus x1*x2*x4 \oplus x1*x2*x5 \oplus x1*x2 \oplus x4*x5 \oplus x4 \oplus x5 \oplus 1$
\item $g = x0*x1*x2*x3*x4*x5$
\end{itemize}
\end{remark2}
\chapter{Linear Cryptanalysis}

\vspace{-25pt}
In this chapter we search for weak long term keys that exhibit linear invariances. In section 4.3 we explain the significance of these keys to mount attacks on T-310.
\vspace{-20pt}
\begin{notes}
\item For the remaining of this chapter we let $X^{(i)}$ denote values inside round $i$.
\end{notes}
\vspace{-30pt}

\section{Weak Keys with One-Bit
Correlations}
\begin{theorem}[A class of Alpha to Alpha properties]
\label{Preconditions6rAlphtoalph}
For each long term KT1 key such that
$D(6)=28$, $D(7)=36$, $D(8)=20$, $P(4)=30$, $P(5)=22$, $P(8)=18$, $P(10)=34$ and for any short term key on 240 bits, and for any initial state on 36 bits, we have the $\alpha\to\alpha$ linear approximation,
$[25]\to[25]$, for 6 rounds.
\end{theorem}
\vspace{-25pt}

\begin{proof}
We will prove the above theorem using the long term key $625$ which has the above characteristics.\\
$625: P=7, 32, 33, 30, 22, 20, 5, 18, 9, 34, 35, 31, 36, 28, 21, 24, 27, 25, 26, 16, 4, 23, 19, 29, 8, 12, 11$ $D=0, 32, 24, 8, 12, 28, 36, 20, 4$\\
We will show that the following holds:

\begin{table}[h]
 \caption{A Detailed Explanation for Key 625}
\label{SomeGoodLC1to1Invariant625details}
\begin{center}
{\small \begin{tabular}{|c|c|c|c|}
\hline
rounds & input $\to$ output & bias\\
\hline
3& [25] $\to$ [28] & $2^{- 1.0}$\\
1& [28] $\to$ [19,21,29,35] & $2^{- 3.4}$\\
2& [19,21,29,35] $\to$ [25] & $2^{- 2.4}$\\
\hline
\end{tabular}}
\end{center}
\end{table}

\begin{figure}[h]
    \centering
    \includegraphics[scale=0.47]{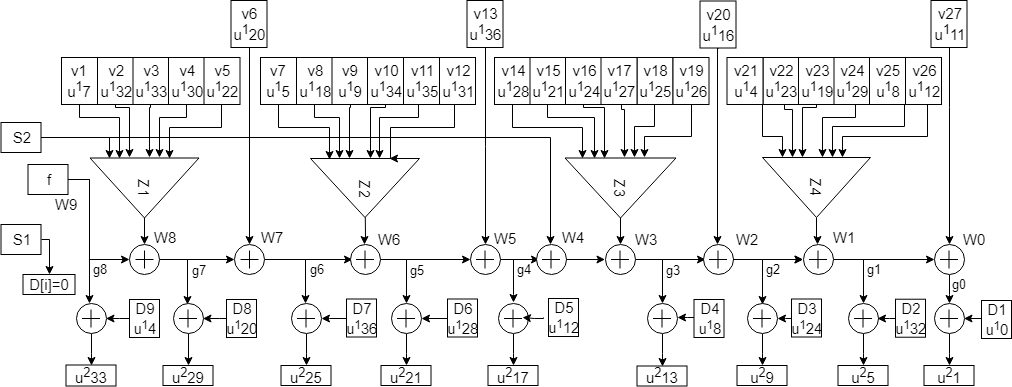}
    \caption{One Round of T-310 for Key 625}
    \label{fig:mesh1}
\end{figure}

First, we observe that $[25]\to[26]\to[27]\to[28]$ for $3$ rounds. So $u_{25}^{(1)}=u_{28}^{(3)}$.\\
We combine equations (6) and (7) to get
\begin{equation} 
\label{eq:3.1}
U_6=u_{D(6)}\oplus Z_2(v7-v12)\oplus U_8 \oplus u_{D(8)} \oplus u_{P(6)}
\end{equation}
From the description of KT1 keys we have that $P(6)=D(8)$. This means that $v6=u_{20}^{(3)}=u_{D(8)}$ and the XOR of two quantities that are equal is zero. Furthermore, according to the theorem we have $D(6)=28$. So equation \ref{eq:3.1} becomes
\begin{equation} 
\label{eq:3.2}u_{21}^{(4)}\oplus u_{29}^{(4)}\oplus u_{28}^{(3)}=Z_2^{(3)}(v7-v12)
\end{equation}
From the theorem we observe that $v8=u_{18}^{(3)}$ and $v10=u_{34}^{(3)}$ which are two of the inputs of $Z_2$. If we add $u_{18}^{(3)}$ and $u_{34}^{(3)}$ in both sides of the equation \ref{eq:3.2}, we have
\begin{equation} 
\label{eq:3.3}
Z_2^{(3)}(v7-v12)\oplus u_{18}^{(3)}\oplus u_{34}^{(3)}=u_{28}^{(3)}\oplus u_{21}^{(4)}\oplus u_{29}^{(4)}\oplus u_{18}^{(3)}\oplus u_{34}^{(3)}
\end{equation}
We know that bits $u_{18}$ and $u_{34}$ in round $3$ will become bits $u_{19}$ and $u_{35}$ in round $4$, respectively. So equation \ref{eq:3.3} becomes
$$Z_2^{(3)}(v7-v12))\oplus u_{18}^{(3)}\oplus u_{34}^{(3)}=u_{28}^{(3)}\oplus u_{21}^{(4)}\oplus u_{29}^{(4)}\oplus u_{19}^{(4)}\oplus u_{35}^{(4)}$$
and the expression $u_{28}^{(3)}\oplus u_{21}^{(4)}\oplus u_{29}^{(4)}\oplus u_{19}^{(4)}\oplus u_{35}^{(4)}$ is biased.\\
Thus, we have showed that $[28]\to[19,21,29,35]$.\\
In the fifth round the bits $u_{19}$, $u_{21}$, $u_{29}$, $u_{35}$ become $u_{20}$, $u_{22}$, $u_{30}$, $u_{36}$.
According to the theorem, we have $D(7)=36$, $D(8)=20=P(6)$, $P(4)=30$ and $P(5)=22$. So we have $v4=u_{30}^{(5)}$, $v5=u_{22}^{(5)}$, $v6=u_{20}^{(5)}$ and $u_{D(7)}^{(5)}=u_{36}^{(5)}$.\\
Combining equations (7), (8) and (9) we get
\vspace{-5pt}
$$U_7=u_{D(7)}\oplus u_{P(6)}\oplus Z_1(s_2,v1-v5)\oplus f$$
\vspace{-5pt}
which becomes
$$u_{25}^{(6)}=u_{36}^{(5)}\oplus u_{20}^{(5)}\oplus Z_1^{(5)}(s_2,v1-v5)\oplus f^{(5)}$$
We observe that the output of $Z_1^{(5)}$ is correlated to the XOR of two of its inputs $u_{22}^{(5)}\oplus u_{30}^{(5)}$
Hence, we have
$$u_{25}^{(6)}\oplus f^{(5)}=u_{36}^{(5)}\oplus u_{20}^{(5)}\oplus Z_1^{(5)}(s_2,v1-v5)\oplus u_{22}^{(5)}\oplus u_{30}^{(5)}$$
where the right hand side is biased.\\
Thus, we showed that $[19,21,29,35]\to[25]$ for 2 rounds.\\
Therefore, we have shown that $[25]\to[25]$ holds for 6 rounds and all the conditions of the theorem were satisfied.
\end{proof}
\noindent\emph{Proof produced by software\footnote{This is a product of software written by Nicolas T. Courtois}:}
{\small
\begin{lstlisting}
$[25]\to[26]$ i25=o26
$[26]\to[27]$ i26=o27
$[27]\to[28]$ i27=o28
$[28]\to[19,21,29,35]$ Z2=g6$\oplus$g5 Z2=Z2e2+Z2e4 Z2e2=i18 Z2e4=i34 v6=g7$\oplus$g6 v6=i20 
               d8=g7$\oplus$o29 d8=i20 d6=g5$\oplus$o21 d6=i28 i18=o19 i34=o35
$[19,21,29,35]\to[20,22,30,36]$ i19=o20 i21=o22 i29=o30 i35=o36
$[20,22,30,36]\to[25]$ Z1=f$\oplus$g7 Z1=Z1e5+Z1e6 Z1e5=i30 Z1e6=i22 v6=g7$\oplus$g6 v6=i20  
               d7=g6$\oplus$o25 d7=i36

\end{lstlisting}}
\noindent\emph{Another Proof for Key 729:}\\
$729: P=7, 23, 33, 16, 31, 4, 5, 1, 9, 12, 14, 13, 36, 8, 21, 3, 24, 25, 32, 20, 2, 6, 30, 29, 28, 26, 18$ \\$D=0, 12, 16, 28, 8, 32, 36, 4, 24$
{\small
\begin{lstlisting}
$[29]\to[30]$ i29=o30
$[30]\to[31]$ i30=o31
$[31]\to[32]$ i31=o32
$[32]\to[2,21,29]$ Z2=g6$\oplus$g5 Z2=Z2e2 Z2e2=i1 v6=g7$\oplus$g6 v6=i4 d8=g7$\oplus$o29 d8=i4 
            d6=i32 d6=g5$\oplus$o21 i1=o2
$[2,21,29]\to[3,22,30]$ i2=o3 i21=o22 i29=o30 
$[3,22,30]\to[4,23,31]$ i3=o4 i22=o23 i30=o31
$[4,23,31]\to[29]$ Z1=f$\oplus$g7 Z1=Z1e3+Z1e6 Z1e3=i23 Z1e6=i31 d8=g7$\oplus$o29 d8=i4

\end{lstlisting}}

\begin{remark}
Whenever a long term KT1 key exhibits an one-bit correlation of type $[4\cdot k +1] \to [4\cdot k +1]$, where $k \in \{0,1,\dots,7\}$, then the same key can also be used for one-bit correlations of type $[4\cdot k +i] \to [4\cdot k +i]$, where $i \in \{1,2,3,4\}$.
\end{remark}

\begin{conjecture}
From our observations, we believe that the smaller the number $\alpha$ is the hardest it is to produce one-bit correlations. That is because as the number $\alpha$ becomes smaller, the Hamming weight (HW) at round 4 becomes larger. For example, as you can also observe in Appendix C, if $\alpha \in \{29,\dots,32\}$, $HW = 3$ at round 4 and  if $\alpha \in \{25,\dots,29\}$, $HW = 4$ at round 4 etc.
\end{conjecture}

\section{More Weak Keys}
From the theorem below, we can see that it is possible for a long term KT1 key to exhibit an invariant property which includes more than one bits.
\begin{theorem}[A class of 8R properties]
\label{Preconditions6rAlphtoalph}
For each long term KT1 key such that
$\{D(7),D(9)\}\in \{12,16\}$, $\{D(3)/D(4)\}=32$, with the remaining of $\{D(3)/D(4)\}\in\{28,36\}$ and finally the remaining of $\{28,36\}\in P(20)$  and for any short term key on 240 bits, and for any initial state on 36 bits,  
we have the linear approximation
$[9,13]\to[9,13]$ which is true with probability exactly 1.0 for 8 rounds.
\end{theorem}
\vskip1pt

\begin{proof}
A long term key that has the above characteristics is key 788:\\
$788 : P = 26,19,33,36,4,20,5,27,9,17,2,11,12,31,21,22,1,25,7,28,16,24,32,29,8,30,34$\\
$D = 0,4,36,32,24,8,12,20,16$\\
We will show that the following holds:

\begin{table}[h]
\caption{A Detailed Explanation of the 8 Round Property}
\label{SomeGoodLC1to1Invariant625details}
\begin{center}
{\small \begin{tabular}{|c|c|c|c|c|}
\hline
rounds & input $\to$ output & bias\\
\hline
3& [9,13] $\to$ [12,16,$f^{(3)}$] & $2^{- 1.0}$\\
1& [12,16,$f^{(3)}$] $\to$ [25,29,33] & $2^{- 1.0}$\\
3& [25,29,33] $\to$ [28,32,36] & $2^{- 1.0}$\\
1& [28,32,36] $\to$ [9,13] & $2^{- 1.0}$\\
\hline
\end{tabular}}
\end{center}
\end{table}

First, we observe that $[9]\to[10]\to[11]\to[12]$ and $[13]\to[14]\to[15]\to[16]$ for $3$ rounds. So $u_{9}^{(1)}=u_{12}^{(3)}$ and $u_{13}^{(1)}=u_{16}^{(3)}$.\\
From the description of KT1 keys we have that $P(6)=D(8)$. This means that $v6=u_{20}^{(3)}=u_{D(8)}$ and the XOR of two quantities that are equal, is zero. So equation (7) becomes
$$u_{21}^{(4)}\oplus u_{29}^{(4)}\oplus u_{25}^{(4)}=u_{D(7)}^{(3)}$$
\begin{figure}[h]
    \centering
    \includegraphics[scale=0.47]{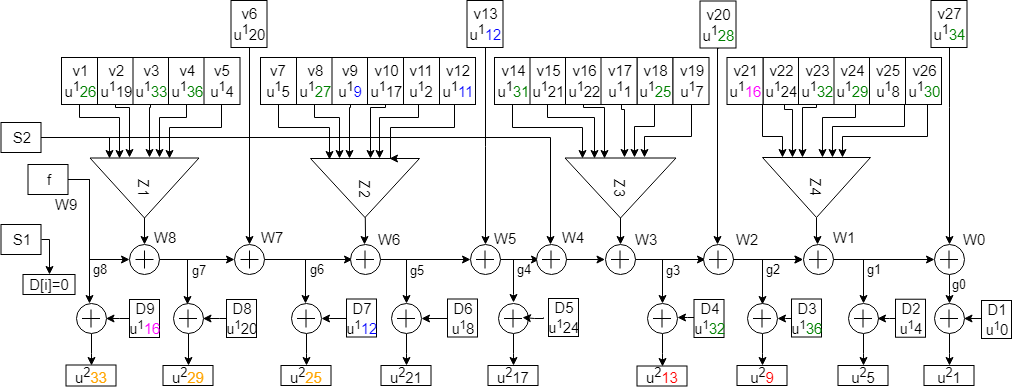}
    \caption{One Round of T-310 for Key 788}
    \label{fig:mesh1}
\end{figure}
\pagebreak
\\Then we have two cases:\\
Case A: If $D(7)=12$ and $D(9)=16$, then we have $u_{29}^{(4)}\oplus u_{25}^{(4)}=u_{12}^{(3)}$ and using (9) we have $u_{33}^{(4)} = f^{(3)}\oplus u_{16}^{(3)}$ which implies that $[12]\to [25,29]$ and $[16,f^{(3)}]\to [33]$.\\
Case B: If $D(7)=16$ and $D(9)=12$, then we have $u_{29}^{(4)}\oplus u_{25}^{(4)}=u_{16}^{(3)}$ and using (9) we have $u_{33}^{(4)} = f^{(3)}\oplus u_{12}^{(3)}$ which implies that $[16]\to [25,29]$ and $[12,f^{(3)}]\to [33]$.\\
So in both cases we have $[12,16,f^{(3)}]\to [25,29,33]$\\
Then we can easily observe that $[25]\to [28]$, $[29]\to [32]$ and $[33]\to [36]$ for 3 rounds.\\
From equation (3), the following holds
$$u_{9}^{(8)}\oplus u_{13}^{(8)}=u_{D(3)}^{(7)}\oplus u_{D(4)}^{(7)}\oplus u_{P(20)}^{(7)}$$
If we first have $\{D(3)/D(4)\}=32$, with the remaining of $\{D(3)/D(4)\}\in\{28,36\}$ and finally the remaining of $\{28,36\}\in P(20)$, then clearly $\{D(3),D(4),P(20)\}\in \{28,32,36\}$. Hence, the following holds
$$u_{9}^{(8)}\oplus u_{13}^{(8)}=u_{28}^{(7)}\oplus u_{32}^{(7)}\oplus u_{36}^{(7)}$$
Therefore, we have successfully shown that $[28,32,36]\to [9,13]$.\\
However, we cannot have $P(20)=32$ and $\{D(3),D(4)\}\in \{28,36\}$ because the following condition of KT1 keys will no longer hold:\\
"There exist $\{j_1,j_2,\dots,j_7,j_8\}$ a permutation of $\{2,3,\dots ,9\}$ which defines\\
$D(i)$ for every $i \in \{2,3,\dots ,9\}$ as follows:\\
$D(j_1)=4,D(j_2)=4j_1,D(j_3)=4j_2,\dots,D(j_8)=4j_7$"
\end{proof}
\vspace{30pt}
However, even if we can construct KT1 keys which have linear invariances such the one we have just studied, in order to mount an attack like the one described in section 4.3.2 we need to find invariant properties that include both IV bits and secret key bits. As we will see in the rest of this section this is also possible.

\newpage
\begin{theorem}[A class of 6R properties]
\label{Preconditions6rAlphtoalph}
For each long term KT1 key such that
$D(7)=16$, $\{D(3),D(4),P(20)\}\subset \{4,8,36\}$, $P(27)=10$ and finally $\{D(2),D(9)\}\subset \{28,32\}$ 
and for any short term key on 240 bits, and for any initial state on 36 bits, we have the linear approximation
$[1,5,15,33]\to[1,5,15,33]$ which is true with probability exactly 1.0 for 6 rounds.
\end{theorem}

\begin{proof}
A long term key that has the above characteristics is key 706:\\
$706 : P = 8,2,33,4,13,20,5,14,9,22,30,31,16,19,21,32,3,25,28,36,27,11,23,29,12,24,10$\\
$D = 0,28,8,4,24,12,16,20,32 $\\
We will show that the following holds:

\begin{table}[h]
\caption{A Detailed Explanation of the 6 Round Property}
\label{SomeGoodLC1to1Invariant625details}
\begin{center}
{\small \begin{tabular}{|c|c|c|c|c|}
\hline
rounds & input $\to$ output & bias\\
\hline
2& [\textcolor{Red}{1},\textcolor{Red}{5},\textcolor{ForestGreen}{15},\textcolor{Red}{33}] $\to$ [\textcolor{Red}{3},\textcolor{Red}{7},\textcolor{ForestGreen}{25},\textcolor{ForestGreen}{29},\textcolor{Red}{35}] & $2^{- 1.0}$\\
2& [\textcolor{Red}{3},\textcolor{Red}{7},\textcolor{ForestGreen}{25},\textcolor{ForestGreen}{29},\textcolor{Red}{35}] $\to$ [\textcolor{Red}{9},\textcolor{Red}{13},\textcolor{ForestGreen}{27},\textcolor{ForestGreen}{31}] & $2^{- 1.0}$\\
2& [\textcolor{Blue}{9},\textcolor{Red}{13},\textcolor{Blue}{27},\textcolor{Blue}{31}] $\to$ [\textcolor{Blue}{1},\textcolor{Blue}{5},\textcolor{Red}{15},\textcolor{Blue}{33}] & $2^{- 1.0}$\\
\hline
\end{tabular}}
\end{center}
\end{table}
\vskip-27pt

First of all, we observe that $[\textcolor{Red}{1}]\to[\textcolor{Red}{3}]$, $[\textcolor{Red}{5}]\to[\textcolor{Red}{7}]$ and $[\textcolor{Red}{33}]\to[\textcolor{Red}{35}]$ for 2 rounds. We also see that $[\textcolor{ForestGreen}{15}]\to[\textcolor{ForestGreen}{16}]$ for 1 round. So $u_1^{(1)}=u_3^{(3)}$, $u_5^{(1)}=u_7^{(3)}$, $u_{33}^{(1)}=u_{35}^{(3)}$ and $u_{15}^{(1)}=u_{16}^{(2)}$.\\
\begin{figure}[h]
    \centering
    \includegraphics[scale=0.24]{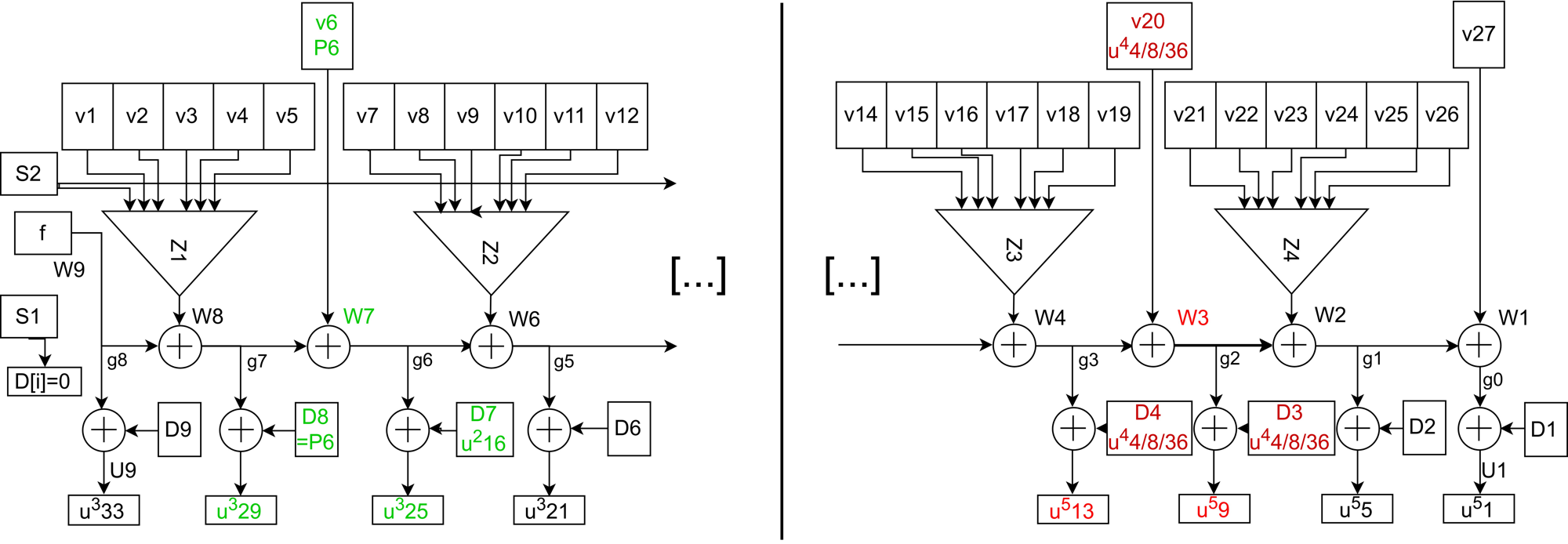}
    \caption{Explanation for our Proof}
    \label{fig:mesh1}
\end{figure}
\\From the description of KT1 keys we have $P(6)=D(8)$ and from the theorem we also have $D(7)=16$. Hence, equation (7) becomes
$$u_{25}^{(3)}\oplus u_{29}^{(3)}=u_{16}^{(2)}$$
Thus, we have $[\textcolor{ForestGreen}{16}]\to[\textcolor{ForestGreen}{25},\textcolor{ForestGreen}{29}]$ for 1 round and, therefore, $[\textcolor{Red}{1},\textcolor{Red}{5},\textcolor{ForestGreen}{15},\textcolor{Red}{33}] \to [\textcolor{Red}{3},\textcolor{Red}{7},\textcolor{ForestGreen}{25},\textcolor{ForestGreen}{29},\textcolor{Red}{35}]$ for 2 rounds.\\
Then we observe that $u_{25}^{(3)}=u_{27}^{(5)}$, $u_{29}^{(3)}=u_{31}^{(5)}$, $u_{3}^{(3)}=u_{4}^{(4)}$, $u_{7}^{(3)}=u_{8}^{(4)}$ and $u_{35}^{(3)}=u_{36}^{(4)}$. According to the theorem, we have the condition $\{D(3),D(4),P(20)\}\subset \{4,8,36\}$. Therefore, equation (3) becomes
$$u_{4}^{(4)}\oplus u_{8}^{(4)}\oplus u_{36}^{(4)} = u_{9}^{(5)}\oplus u_{13}^{(5)}$$
Thus, we have shown that $[\textcolor{Red}{3},\textcolor{Red}{7},\textcolor{ForestGreen}{25},\textcolor{ForestGreen}{29},\textcolor{Red}{35}] \to [\textcolor{Red}{9},\textcolor{Red}{13},\textcolor{ForestGreen}{27},\textcolor{ForestGreen}{31}]$ for 2 rounds.\\
\begin{figure}[h]
    \centering
    \includegraphics[scale=0.4]{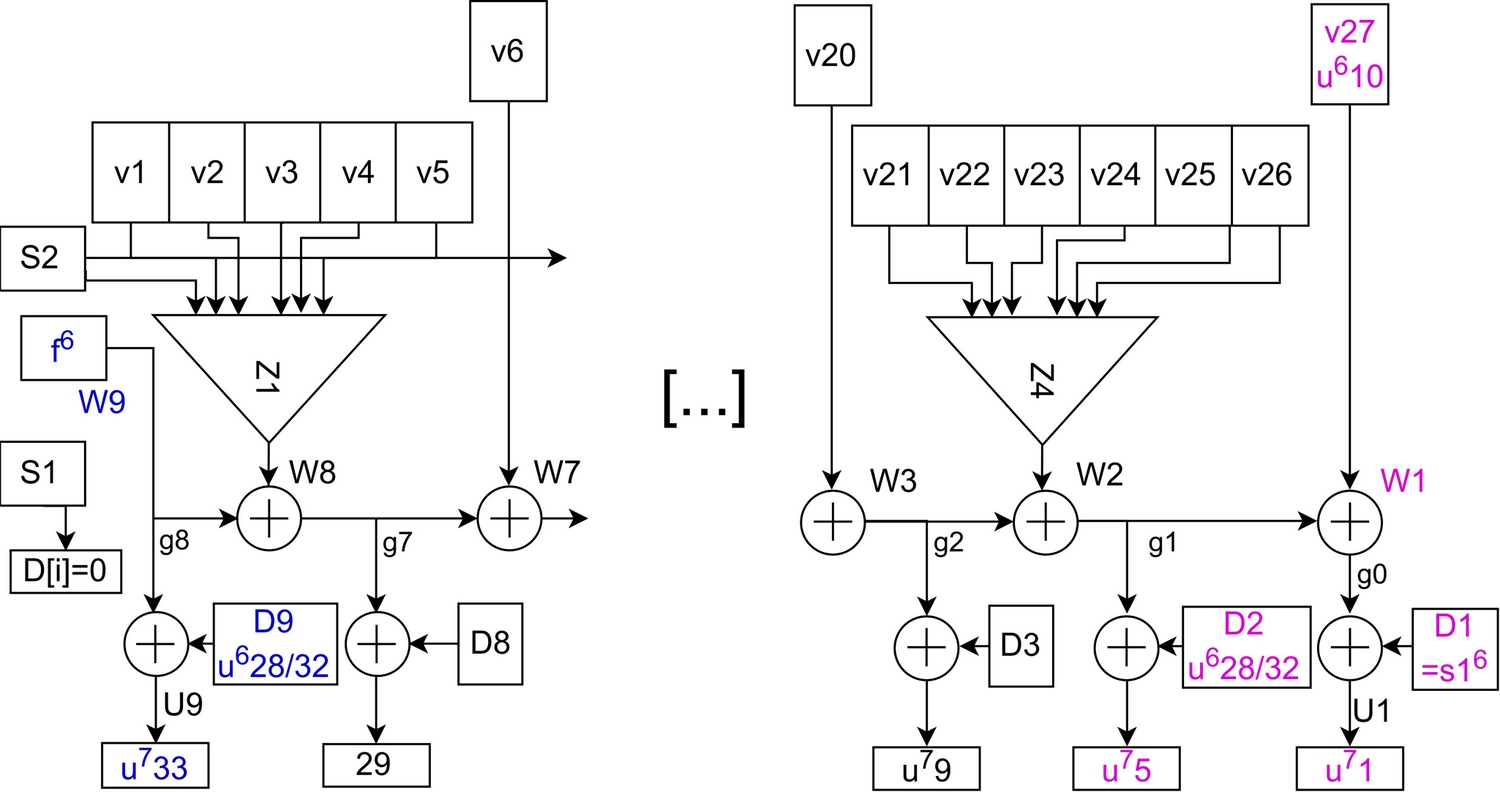}
    \caption{Further Explanation for our Proof}
    \label{fig:mesh1}
\end{figure}
\\We proceed to the final part of the proof as follows:\\
It is clear that $u_{13}^{(5)}=u_{15}^{(7)}$, $u_{9}^{(5)}=u_{10}^{(6)}$, $u_{27}^{(5)}=u_{28}^{(6)}$, and $u_{31}^{(5)}=u_{32}^{(6)}$. The remaining conditions from the theorem are $P(27)=10$ and $\{D(2),D(9)\}\subset \{28,32\}$. Hence, equation (1) becomes
$$u_{10}^{(6)}\oplus s_1^{(6)}\oplus u_{D(2)}^{(6)}=u_{1}^{(7)} \oplus u_{5}^{(7)}$$
and equation (9) becomes 
$$u_{D(9)}^{(6)}\oplus u_{33}^{(7)}=f^{(6)}$$
At this final step, we see that if $\{D(2),D(9)\}\subset \{28,32\}$ we have $[\textcolor{Blue}{10},\textcolor{Blue}{28},\textcolor{Blue}{32}]\to [\textcolor{Blue}{1},\textcolor{Blue}{5},\textcolor{Blue}{33}]$ for one round. So, we have also shown that $[\textcolor{Blue}{9},\textcolor{Red}{13},\textcolor{Blue}{27},\textcolor{Blue}{31}] \to [\textcolor{Blue}{1},\textcolor{Blue}{5},\textcolor{Red}{15},\textcolor{Blue}{33}]$ for 2 rounds.\\
Therefore, if the conditions of the theorem are satisfied , we have
$[\textcolor{Red}{1},\textcolor{Red}{5},\textcolor{ForestGreen}{15},\textcolor{Red}{33}]\to[\textcolor{Blue}{1},\textcolor{Blue}{5},\textcolor{Red}{15},\textcolor{Blue}{33}]$ for 6 rounds.
\end{proof}

\begin{theorem}[A class of 4R properties]
\label{Preconditions4rAlphtoalph}
For each long term KT1 key such that
$D(2)=36$, $D(9)=4$, $\{D(5),D(6),D(7)\}\subset \{8,20,24\}$ and $P(27)=7$ and for any short term key on 240 bits, and for any initial state on 36 bits, we have the linear approximation
$[4,8,19-20,23-24,36,s_1^{(1)},f^{(1)}]\to[4,8,19-20,23-24,36]$ which is true with probability exactly 1.0 for 4 rounds.
\end{theorem}

\vspace{-16pt}
\begin{proof}
We will show that the following holds:

\begin{table}[h]
\caption{A Detailed Explanation of the 4 Round Property}
\label{SomeGoodLC1to1Invariant625details}
\begin{center}
{\small \begin{tabular}{|c|c|c|c|c|}
\hline
rounds & input $\to$ output & bias\\
\hline
1& $[\textcolor{ForestGreen}{4},\textcolor{Blue}{8},\textcolor{Brown}{19}-\textcolor{Blue}{20},\textcolor{Brown}{23}-\textcolor{Blue}{24},\textcolor{Red}{36},\textcolor{Red}{s_1^{(1)}},\textcolor{ForestGreen}{f^{(1)}}] \to [\textcolor{Red}{1},\textcolor{Red}{5},\textcolor{Red}{8},\textcolor{Blue}{17},\textcolor{Brown}{20}-\textcolor{Blue}{21},\textcolor{Brown}{24},\textcolor{ForestGreen}{33}]$ & $2^{- 1.0}$\\
1& $[\textcolor{Red}{1},\textcolor{Red}{5},\textcolor{Brown}{8},\textcolor{Blue}{17},\textcolor{Brown}{20}-\textcolor{Blue}{21},\textcolor{Brown}{24},\textcolor{ForestGreen}{33}] \to [\textcolor{Red}{2},\textcolor{Red}{6},\textcolor{Brown}{17}-\textcolor{Blue}{18},\textcolor{Brown}{21}-\textcolor{Blue}{22},\textcolor{ForestGreen}{34}]$ & $2^{- 1.0}$\\
2& $[\textcolor{Red}{2},\textcolor{Red}{6},\textcolor{Brown}{17}-\textcolor{Blue}{18},\textcolor{Brown}{21}-\textcolor{Blue}{22},\textcolor{ForestGreen}{34}] \to [\textcolor{Red}{4},\textcolor{Red}{8},\textcolor{Brown}{19}-\textcolor{Blue}{20},\textcolor{Brown}{23}-\textcolor{Blue}{24},\textcolor{ForestGreen}{36}]$ & $2^{- 1.0}$\\
\hline
\end{tabular}}
\end{center}
\end{table}

We observe that $u_{19}^{(1)}=u_{20}^{(2)}$ and $u_{23}^{(1)}=u_{24}^{(2)}$.\\
According to the theorem, $D(2)=36$, $P(27)=7$ and using equation (1)  we have
$$u_1^{(2)} \oplus s_{1}^{(1)}= u_5^{(2)} \oplus u_{36}^{(1)}  \oplus u_{7}^{(1)}$$
However we know that $u_{7}^{(1)}=u_{8}^{(2)}$ so the previous equation becomes 
$$u_1^{(2)} \oplus u_5^{(2)} \oplus u_{8}^{(2)}= u_{36}^{(1)}  \oplus s_{1}^{(1)} $$
Hence, $[\textcolor{Red}{36},\textcolor{Red}{s_1^{(1)}}] \to [\textcolor{Red}{1},\textcolor{Red}{5},\textcolor{Red}{8}]$.\\
From the theorem, we also have the condition $\{D(5),D(6),D(7)\}\subset \{8,20,24\}$. From the description of KT1 keys we have that $P(13) = D(7)$ and ,thus, equation (5) becomes
$$u_{17}^{(2)} \oplus u_{21}^{(2)} = u_{8}^{(1)} \oplus u_{20}^{(1)}  \oplus u_{24}^{(1)}$$
Hence, $[\textcolor{Blue}{8},\textcolor{Blue}{20},\textcolor{Blue}{24}] \to [\textcolor{Blue}{17},\textcolor{Blue}{21}]$.\\
The last condition of the theorem is $D(9)=4$ and equation (9) becomes
$$u_{33}^{(2)} \oplus u_{4}^{(1)} = f^{(1)}$$
Hence, $[\textcolor{ForestGreen}{4},\textcolor{ForestGreen}{f^{(1)}}] \to [\textcolor{ForestGreen}{33}]$.\\
Combining all the previous results, we have $[\textcolor{ForestGreen}{4},\textcolor{Blue}{8},\textcolor{Brown}{19}-\textcolor{Blue}{20},\textcolor{Brown}{23}-\textcolor{Blue}{24},\textcolor{Red}{36},\textcolor{Red}{s_1^{(1)}},\textcolor{ForestGreen}{f^{(1)}}] \to [\textcolor{Red}{1},\textcolor{Red}{5},\textcolor{Red}{8},\textcolor{Blue}{17},\textcolor{Brown}{20}-\textcolor{Blue}{21},\textcolor{Brown}{24},\textcolor{ForestGreen}{33}]$ for 1 round.\\
In the second round we observe that $u_{1}^{(2)}=u_{2}^{(3)}$, $u_{5}^{(2)}=u_{6}^{(3)}$, $u_{17}^{(2)}=u_{18}^{(3)}$, $u_{21}^{(2)}=u_{22}^{(3)}$, $u_{33}^{(2)}=u_{34}^{(3)}$.\\
Using equation (5) again we deduce that $[\textcolor{Brown}{8},\textcolor{Brown}{20},\textcolor{Brown}{24}] \to [\textcolor{Brown}{17},\textcolor{Brown}{21}]$.\\
Thus, we have $[\textcolor{Red}{1},\textcolor{Red}{5},\textcolor{Brown}{8},\textcolor{Blue}{17},\textcolor{Brown}{20}-\textcolor{Blue}{21},\textcolor{Brown}{24},\textcolor{ForestGreen}{33}] \to [\textcolor{Red}{2},\textcolor{Red}{6},\textcolor{Brown}{17}-\textcolor{Blue}{18},\textcolor{Brown}{21}-\textcolor{Blue}{22},\textcolor{ForestGreen}{34}]$ for 1 round.\\
Finally, it is clear that $[\textcolor{Red}{2},\textcolor{Red}{6},\textcolor{Brown}{17}-\textcolor{Blue}{18},\textcolor{Brown}{21}-\textcolor{Blue}{22},\textcolor{ForestGreen}{34}] \to [\textcolor{Red}{4},\textcolor{Red}{8},\textcolor{Brown}{19}-\textcolor{Blue}{20},\textcolor{Brown}{23}-\textcolor{Blue}{24},\textcolor{ForestGreen}{36}]$ for 2 rounds and this completes the proof.
\end{proof}
\vspace{-15pt}
A very similar theorem is:
\begin{theorem}[A class of 4R properties]
\label{Preconditions4rAlphtoalph}
For each long term KT1 key such that
$D(2)=36$, $D(9)=4$, $\{D(5),D(6),D(7)\}\subset \{8,20,24\}$ and $P(27)=6$ and for any short term key on 240 bits, and for any initial state on 36 bits, we have the linear approximation\\
$[4,8,18,20,22,24,36,s_1^{(1)},f^{(1)}]\to[4,8,18,20,22,24,36]$ which is true with probability exactly 1.0 for 4 rounds.
\end{theorem}

\begin{proof}
We will show that the following holds:

\begin{table}[h]
\caption{A Detailed Explanation of the 4 Round Property}
\label{SomeGoodLC1to1Invariant625details}
\begin{center}
{\small \begin{tabular}{|c|c|c|c|c|}
\hline
rounds & input $\to$ output & bias\\
\hline
1& $[\textcolor{ForestGreen}{4},\textcolor{Blue}{8},\textcolor{Brown}{18},\textcolor{Blue}{20},\textcolor{Brown}{22},\textcolor{Blue}{24},\textcolor{Red}{36},\textcolor{Red}{s_1^{(1)}},\textcolor{ForestGreen}{f^{(1)}}] \to [\textcolor{Red}{1},\textcolor{Red}{5},\textcolor{Red}{7},\textcolor{Blue}{17},\textcolor{Brown}{19},\textcolor{Blue}{21},\textcolor{Brown}{23},\textcolor{ForestGreen}{33}]$ & $2^{- 1.0}$\\
1& $[\textcolor{Red}{1},\textcolor{Red}{5},\textcolor{Red}{7},\textcolor{Blue}{17},\textcolor{Brown}{19},\textcolor{Blue}{21},\textcolor{Brown}{23},\textcolor{ForestGreen}{33}] \to [\textcolor{Red}{2},\textcolor{Red}{6},\textcolor{Red}{8},\textcolor{Blue}{18},\textcolor{Brown}{20},\textcolor{Blue}{22},\textcolor{Brown}{24},\textcolor{ForestGreen}{34}]$ & $2^{- 1.0}$\\
2& $[\textcolor{Red}{2},\textcolor{Red}{6},\textcolor{Brown}{8},\textcolor{Blue}{18},\textcolor{Brown}{20},\textcolor{Blue}{22},\textcolor{Brown}{24},\textcolor{ForestGreen}{34}] \to [\textcolor{Red}{4},\textcolor{Red}{8},\textcolor{Brown}{18},\textcolor{Blue}{20},\textcolor{Brown}{22},\textcolor{Blue}{24},\textcolor{ForestGreen}{36}]$ & $2^{- 1.0}$\\
\hline
\end{tabular}}
\end{center}
\end{table}
\vspace{-10pt}
We recall a subset of equations from section 2.1
\vspace{-20pt}
\begin{align}
U_1 \oplus s_{1}&= U_2 \oplus u_{D(2)}  \oplus u_{P(27)}\tag{1}\\
U_5 \oplus u_{D(5)} &= U_6 \oplus u_{D(6)}  \oplus u_{P(13)}\tag{5}\\
U_9 \oplus u_{D(9)} &= f~~~~~~~~~~~~ &\tag{9}
\end{align}
We observe that $u_{18}^{(1)}=u_{19}^{(2)}$ and $u_{22}^{(1)}=u_{23}^{(2)}$.\\
According to the theorem, $D(2)=36$, $P(27)=6$ and using equation (1)  we have
$$u_1^{(2)} \oplus s_{1}^{(1)}= u_5^{(2)} \oplus u_{36}^{(1)}  \oplus u_{6}^{(1)}$$
However we know that $u_{6}^{(1)}=u_{7}^{(2)}$ so the previous equation becomes 
$$u_1^{(2)} \oplus u_5^{(2)} \oplus u_{7}^{(2)}= u_{36}^{(1)}  \oplus s_{1}^{(1)} $$
Hence, $[\textcolor{Red}{36},\textcolor{Red}{s_1^{(1)}}] \to [\textcolor{Red}{1},\textcolor{Red}{5},\textcolor{Red}{7}]$.\\
From the theorem, we also have $\{D(5),D(6),D(7)\}\subset \{8,20,24\}$. From the description of KT1 keys we have that $P(13) = D(7)$ and ,thus, equation (5) becomes
$$u_{17}^{(2)} \oplus u_{21}^{(2)} = u_{8}^{(1)} \oplus u_{20}^{(1)}  \oplus u_{24}^{(1)}$$
Hence, $[\textcolor{Blue}{8},\textcolor{Blue}{20},\textcolor{Blue}{24}] \to [\textcolor{Blue}{17},\textcolor{Blue}{21}]$.\\
The last condition of the theorem is $D(9)=4$ and equation (9) becomes
$$u_{33}^{(2)} \oplus u_{4}^{(1)} = f^{(1)}$$
Hence, $[\textcolor{ForestGreen}{4},\textcolor{ForestGreen}{f^{(1)}}] \to [\textcolor{ForestGreen}{33}]$.\\
Combining all the previous results, we have $[\textcolor{ForestGreen}{4},\textcolor{Blue}{8},\textcolor{Brown}{18},\textcolor{Blue}{20},\textcolor{Brown}{22},\textcolor{Blue}{24},\textcolor{Red}{36},\textcolor{Red}{s_1^{(1)}},\textcolor{ForestGreen}{f^{(1)}}] \to [\textcolor{Red}{1},\textcolor{Red}{5},\textcolor{Red}{7},\textcolor{Blue}{17},\textcolor{Brown}{19},\textcolor{Blue}{21},\textcolor{Brown}{23},\textcolor{ForestGreen}{33}]$ for 1 round.\\
In the second round it is clear that $[\textcolor{Red}{1},\textcolor{Red}{5},\textcolor{Red}{7},\textcolor{Blue}{17},\textcolor{Brown}{19},\textcolor{Blue}{21},\textcolor{Brown}{23},\textcolor{ForestGreen}{33}] \to [\textcolor{Red}{2},\textcolor{Red}{6},\textcolor{Red}{8},\textcolor{Blue}{18},\textcolor{Brown}{20},\textcolor{Blue}{22},\textcolor{Brown}{24},\textcolor{ForestGreen}{34}]$.\\
In the following two rounds we observe that $u_{2}^{(3)}=u_{4}^{(5)}$, $u_{6}^{(3)}=u_{8}^{(5)}$, $u_{18}^{(3)}=u_{20}^{(5)}$, $u_{22}^{(3)}=u_{24}^{(5)}$, $u_{34}^{(3)}=u_{36}^{(5)}$.\\
Using equation (5) again we deduce that $[\textcolor{Brown}{8},\textcolor{Brown}{20},\textcolor{Brown}{24}] \to [\textcolor{Brown}{17},\textcolor{Brown}{21}]$ for 1 round.\\
Furthermore, we can notice that $u_{17}^{(4)}=u_{18}^{(5)}$, $u_{21}^{(4)}=u_{22}^{(5)}$.\\
Thus, we have $[\textcolor{Red}{2},\textcolor{Red}{6},\textcolor{Brown}{8},\textcolor{Blue}{18},\textcolor{Brown}{20},\textcolor{Blue}{22},\textcolor{Brown}{24},\textcolor{ForestGreen}{34}] \to [\textcolor{Red}{4},\textcolor{Red}{8},\textcolor{Brown}{18},\textcolor{Blue}{20},\textcolor{Brown}{22},\textcolor{Blue}{24},\textcolor{ForestGreen}{36}]$ for 2 rounds and this completes the proof.
\end{proof}

\pagebreak
\begin{theorem}[A class of 8R properties]
\label{Preconditions8rAlphtoalph}
For each long term KT1 key such that
$D(2)=36$, $D(9)=4$, $\{D(5),D(6),D(7)\}\subset \{8,20,24\}$, $P(27)=6$ and for any short term key on 240 bits, and for any initial state on 36 bits, we have the linear approximation
$[1,3,5,17,21]\to[1,3,5,17,21]$ which is true with probability exactly 1.0 for 8 rounds.
\end{theorem}

\vspace{-20pt}
\begin{proof}
We will show that the following holds:

\begin{table}[h]
\caption{A Detailed Explanation of the 8 Round Property}
\label{SomeGoodLC1to1Invariant625details}
\begin{center}
{\small \begin{tabular}{|c|c|c|c|c|}
\hline
rounds & input $\to$ output & bias\\
\hline
1& $[\textcolor{ForestGreen}{1},\textcolor{ForestGreen}{3},\textcolor{Blue}{5},\textcolor{Blue}{17},\textcolor{Blue}{21}] \to [\textcolor{ForestGreen}{2},\textcolor{ForestGreen}{4},\textcolor{Blue}{6},\textcolor{Blue}{18},\textcolor{Blue}{22},\textcolor{ForestGreen}{f^{(2)}}]$ & $2^{- 1.0}$\\
2& $[\textcolor{ForestGreen}{2},\textcolor{ForestGreen}{4},\textcolor{Blue}{6},\textcolor{Blue}{18},\textcolor{Blue}{22},\textcolor{ForestGreen}{f^{(2)}}] \to [\textcolor{ForestGreen}{4},\textcolor{Blue}{8},\textcolor{Blue}{20},\textcolor{Blue}{24},\textcolor{ForestGreen}{34},\textcolor{ForestGreen}{f^{(4)}}]$ & $2^{- 1.0}$\\
2& $[\textcolor{ForestGreen}{4},\textcolor{Blue}{8},\textcolor{Blue}{20},\textcolor{Blue}{24},\textcolor{ForestGreen}{34},\textcolor{ForestGreen}{f^{(4)}}] \to [\textcolor{Blue}{18},\textcolor{Blue}{22},\textcolor{ForestGreen}{34},\textcolor{ForestGreen}{36},\textcolor{Red}{s_1^{(6)}}]$ & $2^{- 1.0}$\\
2& $[\textcolor{Blue}{18},\textcolor{Blue}{22},\textcolor{ForestGreen}{34},\textcolor{Red}{36},\textcolor{Red}{s_1^{(6)}}] \to [\textcolor{Red}{2},\textcolor{Red}{6},\textcolor{Red}{8},\textcolor{Blue}{20},\textcolor{Blue}{24},\textcolor{ForestGreen}{36},\textcolor{RedViolet}{s_1^{(8)}}]$ & $2^{- 1.0}$\\
1& $[\textcolor{Red}{2},\textcolor{RedViolet}{6},\textcolor{Blue}{8},\textcolor{Blue}{20},\textcolor{Blue}{24},\textcolor{RedViolet}{36},\textcolor{RedViolet}{s_1^{(8)}}] \to [\textcolor{RedViolet}{1},\textcolor{Red}{3},\textcolor{RedViolet}{5},\textcolor{Blue}{17},\textcolor{Blue}{21}]$ & $2^{- 1.0}$\\
\hline
\end{tabular}}
\end{center}
\end{table}

In order to explain the results in the above table we will use the following facts:\\
Fact A: It is clear that $u_{4k+j}^{(i)} \to u_{4k+j+1}^{(i+1)}$ where $k \in \{0,\dots,8\}$ and $j \in \{1,2,3\}$, for some rounds $i$ and $i+1$.\\
Fact B: According to the theorem, $D(2)=36$, $P(27)=6$ and using equation (1)  we have
$$u_1^{(i+1)} \oplus s_{1}^{(i)}= u_5^{(i+1)} \oplus u_{36}^{(i)}  \oplus u_{6}^{(i)}$$
So we have $[\textcolor{RedViolet}{6},\textcolor{RedViolet}{36},\textcolor{RedViolet}{s_1^{(i)}}] \to [\textcolor{RedViolet}{1},\textcolor{RedViolet}{5}]$.\\
Fact C: However we know that $u_{6}^{(i)}=u_{7}^{(i+1)}$ so the previous equation can also become 
$$u_1^{(i+1)} \oplus u_5^{(i+1)} \oplus u_{7}^{(i+1)}= u_{36}^{(i)}  \oplus s_{1}^{(i)} $$
Hence, $[\textcolor{Red}{36},\textcolor{Red}{s_1^{(i)}}] \to [\textcolor{Red}{1},\textcolor{Red}{5},\textcolor{Red}{7}]$ also holds.\\
Fact D: From the theorem, we also have $\{D(5),D(6),D(7)\}\subset \{8,20,24\}$. From the description of KT1 keys we have that $P(13) = D(7)$ and ,thus, equation (5) becomes
$$u_{17}^{(i+1)} \oplus u_{21}^{(i+1)} = u_{8}^{(i)} \oplus u_{20}^{(i)}  \oplus u_{24}^{(i)}$$
Hence, $[\textcolor{Blue}{8},\textcolor{Blue}{20},\textcolor{Blue}{24}] \to [\textcolor{Blue}{17},\textcolor{Blue}{21}]$.\\
Fact E: The last condition of the theorem is $D(9)=4$ and equation (9) becomes
\vskip-15pt
$$u_{33}^{(i+1)} \oplus u_{4}^{(i)} = f^{(i)}$$
Hence, $[\textcolor{ForestGreen}{4},\textcolor{ForestGreen}{f^{(i)}}] \to [\textcolor{ForestGreen}{33}]$.\\
Thus, combining some of the above facts in each round, we can prove that $[\textcolor{ForestGreen}{1},\textcolor{ForestGreen}{3},\textcolor{Blue}{5},\textcolor{Blue}{17},\textcolor{Blue}{21}] \to [\textcolor{RedViolet}{1},\textcolor{Red}{3},\textcolor{RedViolet}{5},\textcolor{Blue}{17},\textcolor{Blue}{21}]$ holds for 8 rounds.
\end{proof}

\section{Slide Attacks on T-310}

In this section, we perform two slide attacks on T-310, described in [7,11], which are based on the three assumptions below:

\begin{enumerate}
\item The attacker using a decryption oracle attack, described in chapter 3 of [7], obtains access to $73\%$ of the values $a_i \myeq u_{127\cdot i, \alpha}$, where $i=1,2,3,\cdots$.
\item We recall that the key bits are repeated every 120 rounds. The attacker chooses an IV at random and creates several other IV's which differ by clocking the IV LFSR backwards by $120 \cdot s$ steps, for some integer $s$.
\item The attacker is able to identify whether the internal states are identical,\newline i.e. $u_{0,1-36}=u^{'}_{120\cdot s,1-36}=$ 0xC5A13E396.
\end{enumerate}
We take into account the equation below for the following two attacks:
$$120\cdot s = 127 \cdot t + d$$

\pagebreak
\subsection{One Bit Correlation Attack}
The main purpose of the attacker is to be able to identify whether $u_{0,1-36}=u^{'}_{120\cdot s,1-36}=$ 0xC5A13E396, by using a decryption oracle attack. For example, if $s=19$, $t=18$, $d=-6$ and the bit $\alpha \in \{25,26,27,28\}$, then we can use the long term key 625, to identify if this is the case, since there are correlations between the bits $\alpha$ for 6 rounds. We can also use a simpler case when $s=1$, $t=1$, $d=-7$, in order to mount such an attack, by constructing keys according to the preconditions mentioned in Appendix C.\\
More details on the attack can be found in [7].

\subsection{Attack based on Multiple Bits Invariant Property}
A most recent attack proposed in [11] is based on linear invariant properties that involve some IV and key bits.
\\The attack proceeds as follows:\\
Based on the above assumptions, we let $s=127$, $t=120$ and $d = 0$ and we choose our ciphertexts that are produced with IV to have 5 characters length and those produced with IV', which differ by clocking the IV LFSR backwards by $120 \cdot 127$ steps, to have 1205 characters length. We chose this specific length for our ciphertexts such that after the $120 \cdot 127$ steps, the $5 \cdot 13 = 65$ bits $u_{a \cdot i}$ produced for the encryption of at least 5 characters have an overlap on the secret key bits and IV bits.\\
As mentioned in the assumptions above, we obtain access to $73\%$ of the bits $u_{a \cdot i}$. So for two encryptions that used IV and IV', respectively, we will obtain access to $0.73 \cdot 0.73 \cdot 65 \approx 35$ values $u_{a \cdot i}$ and $u^{'}_{a \cdot i}$ at the same positions $i$.\\
Then we can conclude that the two states are identical, i.e. $u_{0,1-36}=u^{'}_{120\cdot 127,1-36}=$ 0xC5A13E396, if all 35 bits $u_{a \cdot i}$ and 35 bits $u^{'}_{a \cdot i}$ at the same positions $i$ are exactly the same.\\
In the final step of our attack, we exploit this fact and the fact that the long term key, for example a key that satisfies the conditions of Theorem 3.6, has a linear property for 8 rounds that involves IV and key bits (in our case s1) to construct 8 parity equations and, thus, recover 8 bits of the secret key.
\chapter{Generalised Linear Cryptanalysis}

In this chapter we use Generalised Linear Cryptanalysis in order to search for invariant, not linear, properties of long term keys that hold for one round, a technique that was used in [8] to formulate an attack on DES. When searching for non linear invariant properties we face a three dimensional problem. The three dimensions are $(a)$ the choice of a long term key, $(b)$ the choice of a random Boolean function and $(c)$ the choice of the non linear invariant property. So, we can pick two of the three dimensions and search for the answer in the third dimension. In this chapter, we search for the non linear invariant properties occuring due to our choice of the long term key and the choice of a random Boolean function.

\begin{figure}[h]
    \centering
    \includegraphics[scale=0.6]{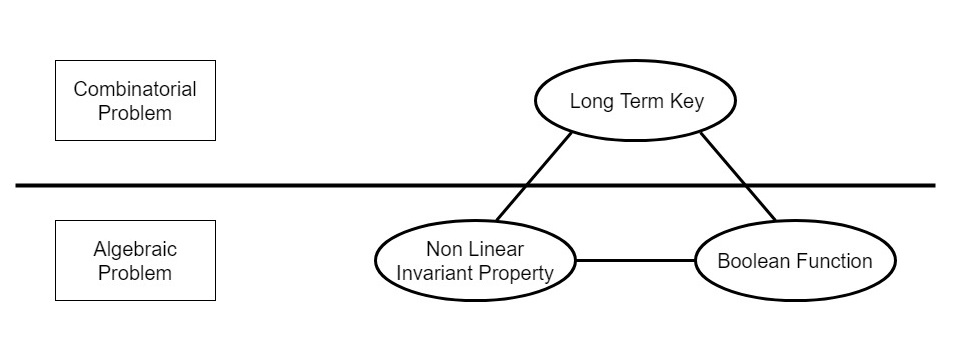}
    \caption{Three Dimensional Problem}
    \label{fig:mesh1}
\end{figure}

\section{Procedure of Finding Invariants on 12 Bits}
As a first step to our research, we searched for invariants on 12 bits, on any long term key, and more precisely on the bits $u_{25}$ up to $u_{36}$. So, the area we focused on is displayed in the figure below.

\begin{figure}[h]
    \centering
    \includegraphics[scale=0.4]{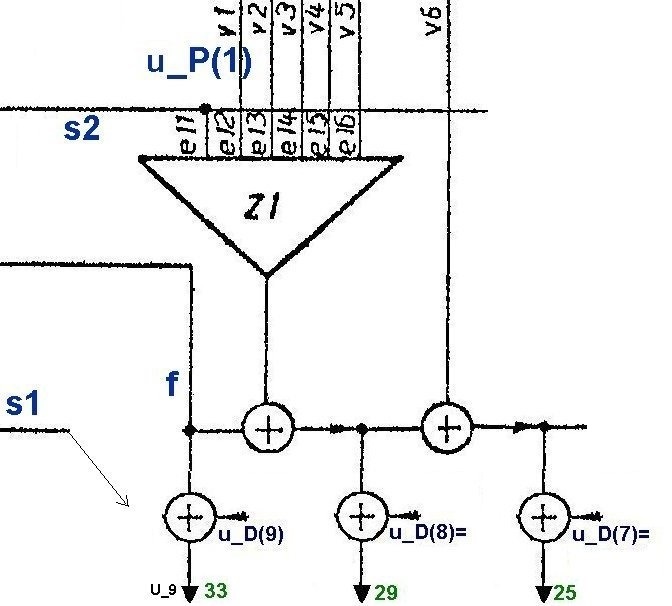}
    \caption{Area of Interest for 12 Bit Invariants}
    \label{fig:mesh1}
\end{figure}

By following a brute force approach in order to search for invariants on the 12 bits mentioned above, we produced some scripts (see Appendix E.2) to help us  follow these 3 steps:\\
\textbf{Step 1:} By renaming each bit $u_i$, where $i \in \{25,26,\dots,36\}$, by the letters $\{l,k,j,\dots,b,a\}$, the IV bit, the $s2$ bit and the Boolean function Z1 by the letters '$F$', '$L$' and '$Z$', respectively, and based on the long term key, we produced a file which contains two columns of data. In the first column, there are all possible monomials on the 12 bits. In the second column, there are the resulting polynomials for each possible monomial after one round.\\
\textbf{Step 2:} In all resulting polynomials, we substituted $F$ by 0 or 1.\\
\textbf{Step 3:} We XORed (exclusive or) each possible monomial with the corresponding polynomial that was produced after one round and computed the kernel of the resulting matrix which will produce the invariant space on the 12 bits.

\vspace{40pt}
\section{Invaribants on 12 Bits}
Below we demonstrate some of the most significant results we have found.

\begin{theorem}[Invariants for $D(9) = 36$] Following the steps mentioned in Section 4.1 and by setting $F=0$ for keys with $D(9) = 36$ we observe the following five invariants for one round:
\begin{spacing}{2}
\begin{itemize}[noitemsep,topsep=0pt]
\item $a \oplus b \oplus c \oplus d \to a \oplus b \oplus c \oplus d$
\item $ac \oplus bd \to ac \oplus bd$
\item $ab \oplus ad \oplus bc \oplus cd \to ab \oplus ad \oplus bc \oplus cd$
\item $abc \oplus abd \oplus acd \oplus bcd \to abc \oplus abd \oplus acd \oplus bcd$
\item $abcd \to abcd$
\end{itemize}
\end{spacing}
\end{theorem}

\vspace{-30pt}
\begin{proof}
We have $d = u_{33}$, $c = u_{34}$, $b = u_{35}$, $a = u_{36}$\\
We observe that $[d] \to [c]$, $[c] \to [b]$ and $[b] \to [a]$ for one round.\\
From equation (9) the following holds:
$$a = F \oplus d$$
So, if $F=0$ then we have $[a] \to [d]$ in the next round.\\
Then it can be easily verified that the above equations will hold for the next round.
\end{proof}

\begin{theorem}[Invariants for $D(9) = 36$ and $L=0$ or $1$]
By following Theorem 4.1, replacing $Z$ by the original Boolean function of T-310 and using the long term key with\\
$D = 1,13,3,2,11,12,32,28,36$ and \\
$P = 35,27,26,34,31,29,25,14,7,22,15,33,8,30,6,10,23,4,24,18,9,20,17,19,16,5,21$\\ we observe that the invariant space has dimension 25 when $L=0$ and dimension 32 when $L=1$.
\end{theorem}
\noindent\emph{Proof.}
We omit the proof because of the extensive length of the invariants.\\

\pagebreak
\begin{theorem}[Invariants for $D(9) = 36$ and $L=1$]
When modifying Theorem 4.2, by replacing $Z$ with a Boolean function we generated at random we observe that the invariant space has dimension 51 when $L=1$.
\end{theorem}
\noindent\emph{Proof.}
We again omit the proof because of the extensive length of the invariants.\\
However, some of the invariants are listed in Appendix D.
The Boolean function we used in order to observe this large number of invariants is:
\begin{align*}
Z(x0,x1,x2,x3,x4,x5) = x1 \oplus x0x1 \oplus x0x3 \oplus x2x5 \oplus x3x5 \oplus x4x5 \oplus x0x1x2\\
\oplus x0x2x5 \oplus x0x3x4 \oplus x0x3x5 \oplus x0x4x5 \oplus x1x2x5 \oplus x1x4x5 \oplus x2x4x5 \oplus x0x1x3x4\\ 
\oplus x0x2x4x5
\oplus x0x3x4x5 \oplus x1x2x3x4 \oplus x1x2x3x5 \oplus x1x2x4x5 \oplus x1x3x4x5\\
\oplus x0x1x2x3x4 \oplus x0x1x2x3x5 \oplus x0x1x3x4x5 \oplus x0x2x3x4x5 \oplus x1x2x3x4x5
\end{align*}

\begin{table}[h]
\caption{Walsh Spectrum and Autocorrelation Spectrum}
\begin{center}
{\small \begin{tabular}{|c|c|lllcc}
\cline{1-2} \cline{6-7}
Walsh Spectrum & Frequency &  &  & \multicolumn{1}{l|}{} & \multicolumn{1}{c|}{Autocorrelation Spectrum} & \multicolumn{1}{c|}{Frequency} \\ \cline{1-2} \cline{6-7} 
-12            & 1         &  &  & \multicolumn{1}{l|}{} & \multicolumn{1}{c|}{-32}                      & \multicolumn{1}{c|}{2}         \\ \cline{1-2} \cline{6-7} 
-8             & 2         &  &  & \multicolumn{1}{l|}{} & \multicolumn{1}{c|}{-24}                      & \multicolumn{1}{c|}{5}         \\ \cline{1-2} \cline{6-7} 
-6             & 5         &  &  & \multicolumn{1}{l|}{} & \multicolumn{1}{c|}{-16}                      & \multicolumn{1}{c|}{6}         \\ \cline{1-2} \cline{6-7} 
-4             & 6         &  &  & \multicolumn{1}{l|}{} & \multicolumn{1}{c|}{-8}                       & \multicolumn{1}{c|}{10}        \\ \cline{1-2} \cline{6-7} 
-2             & 14        &  &  & \multicolumn{1}{l|}{} & \multicolumn{1}{c|}{0}                        & \multicolumn{1}{c|}{17}        \\ \cline{1-2} \cline{6-7} 
0              & 13        &  &  & \multicolumn{1}{l|}{} & \multicolumn{1}{c|}{8}                        & \multicolumn{1}{c|}{11}        \\ \cline{1-2} \cline{6-7} 
2              & 11        &  &  & \multicolumn{1}{l|}{} & \multicolumn{1}{c|}{16}                       & \multicolumn{1}{c|}{9}         \\ \cline{1-2} \cline{6-7} 
4              & 8         &  &  & \multicolumn{1}{l|}{} & \multicolumn{1}{c|}{24}                       & \multicolumn{1}{c|}{2}         \\ \cline{1-2} \cline{6-7} 
6              & 1         &  &  & \multicolumn{1}{l|}{} & \multicolumn{1}{c|}{32}                       & \multicolumn{1}{c|}{1}         \\ \cline{1-2} \cline{6-7} 
8              & 1         &  &  & \multicolumn{1}{l|}{} & \multicolumn{1}{c|}{64}                       & \multicolumn{1}{c|}{1}         \\ \cline{1-2} \cline{6-7} 
12             & 1         &  &  &                       &                                               &                                \\ \cline{1-2}
\end{tabular}}
\end{center}
\end{table}
As we can see from the above tables, we have similar values for the Walsh and Autocorrelation Spectrum as for the original Boolean function Z of T-310 and the frequencies are again focused around 0 which indicates that the above Boolean function exhibits a good level of resistance against linear and differential cryptanalysis.
\newline
\newline
\newline
\begin{conjecture}[Invariants on 12 bits for $F=0$ and $F=1$]
We believe that invariants on 12 bits are impossible for both $F=0$ and $F=1$ simultaneously.
\end{conjecture}
\vskip-10pt
We support the conjecture by the following two examples. We claim that each of the following invariants holds for either $F=0$ or $F=1$.
\begin{example}
The long term key with $D = 1,26,9,34,4,19,28,36,32$ and\\
$P = 30,35,31,25,33,27,5,22,17,29,13,20,1,9,21,3,24,7,6,28,26,2,32,23,18,4,34$ exhibits the following invariant $\mathcal{P}$:\\
$\mathcal{P} = efghijkl \oplus defghijkl \oplus cefghijkl \oplus cdefghijkl \oplus befghijkl \oplus bdefghijkl \oplus bcefghijkl \oplus bcdefghijkl \oplus aefghijkl \oplus adefghijkl \oplus acefghijkl \oplus acdefghijkl \oplus abefghijkl \oplus abdefghijkl \oplus abcefghijkl \oplus abcdefghijkl$\\\\
According to the long term key we have $P(6) = 27 = j$ and we know that $[j] \to [i]$ in one round. By using equations (7), (8) and (9) we have:
$$i = F \oplus Z \oplus l \oplus i$$
$$e = F \oplus d$$
$$a = F \oplus Z \oplus h$$
We compute $\mathcal{P}$ for the next round, we then replace Z by the boolean function of T-310 and set $F=0$. Thus, $\mathcal{P}$ becomes:\\
$\mathcal{P} = abcdefghijkl \oplus bcdefghijkl \oplus acdefghijkl \oplus abdefghijkl \oplus abcdefgijkl \oplus abcdefghijk \oplus cdefghijkl \oplus bdefghijkl \oplus bcdefgijkl \oplus bcdefghijk \oplus adefghijkl \oplus acdefgijkl \oplus acdefghijk \oplus abdefgijkl \oplus abdefghijk \oplus abcdefgijk \oplus defghijkl \oplus cdefgijkl \oplus cdefghijk \oplus bdefgijkl \oplus bdefghijk \oplus bcdefgijk \oplus adefgijkl \oplus adefghijk \oplus acdefgijk \oplus abdefgijk \oplus defgijkl \oplus defghijk \oplus cdefgijk \oplus bdefgijk \oplus adefgijk \oplus defgijk$\\\\
Clearly, we can deduce from the length of the polynomial that it is not an invariant for $F=0$.
However, when we set $F=1$ instead of $F=0$, $\mathcal{P}$ becomes:\\
$\mathcal{P} = abcdefghijkl \oplus bcdefghijkl \oplus acdefghijkl \oplus abdefghijkl \oplus abcefghijkl \oplus cdefghijkl \oplus bdefghijkl \oplus bcefghijkl \oplus adefghijkl \oplus acefghijkl \oplus abefghijkl \oplus defghijkl \oplus cefghijkl \oplus befghijkl \oplus aefghijkl \oplus efghijkl$\\\\
Thus, the invariant holds only for $F=1$.
\end{example}

\begin{example}
The long term key with $D = 0,30,26,6,7,1,28,36,32$ and\\
$P = 26,30,29,31,27,36,5,18,9,15,10,19,28,13,21,32,17,25,14,7,11,3,20,35,34,33,2$ exhibits the following invariant $\mathcal{P}$:\\
$\mathcal{P} = abcdijkl \oplus abcdhijkl \oplus abcdgijkl \oplus abcdghijkl \oplus abcdfijkl \oplus abcdfhijkl \oplus abcdfgijkl \oplus abcdfghijkl \oplus abcdeijkl \oplus abcdehijkl \oplus abcdegijkl \oplus abcdeghijkl \oplus abcdefijkl \oplus abcdefhijkl \oplus abcdefgijkl \oplus abcdefghijkl$\\\\
According to the equations (7), (8) and (9) and the long term key in which $D(7) = 28$, 
$D(8) = P(6) = 36$, $D(9) = 32$, we have: 
$$i = l \oplus h$$
$$e = F \oplus d$$
$$d = F \oplus Z \oplus h$$
We compute $\mathcal{P}$ for the next round, we then replace Z by the boolean function of T-310 and set $F=0$. Thus, $\mathcal{P}$ becomes:\\
$\mathcal{P} = abcdefghijkl \oplus abcefghijkl \oplus abcdfghijkl \oplus abcdeghijkl \oplus abcdefhijkl \oplus abcdefghijk \oplus abcfghijkl \oplus abceghijkl \oplus abcefhijkl \oplus abcefghijk \oplus abcdghijkl \oplus abcdfhijkl \oplus abcdfghijk \oplus abcdehijkl \oplus abcdeghijk \oplus abcdefhijk \oplus abcghijkl \oplus abcfhijkl \oplus abcfghijk \oplus abcehijkl \oplus abceghijk \oplus abcefhijk \oplus abcdhijkl \oplus abcdghijk \oplus abcdfhijk \oplus abcdehijk \oplus abchijkl \oplus abcghijk \oplus abcfhijk \oplus abcehijk \oplus abcdhijk \oplus abchijk$\\\\
We can easily deduce that it is not an invariant for $F=0$ since the length of the resulting polynomial is bigger than expected.\\
However, when we set $F=1$ instead of $F=0$, $\mathcal{P}$ becomes:\\
$\mathcal{P} = abcdefghijkl \oplus abcdfghijkl \oplus abcdeghijkl \oplus abcdefhijkl \oplus abcdefgijkl \oplus abcdghijkl \oplus abcdfhijkl \oplus abcdfgijkl \oplus abcdehijkl \oplus abcdegijkl \oplus abcdefijkl \oplus abcdhijkl \oplus abcdgijkl \oplus abcdfijkl \oplus abcdeijkl \oplus abcdijkl$\\
Therefore, the invariant holds only for $F=1$.
\end{example}

\vspace{-30pt}
\section{Invariants on 20 Bits}
In this section, we rename each bit $u_i$, where $i \in \{17,18,\dots,36\}$, by the letters $\{t,s,r,\dots,b,a\}$, the IV bit, the s2 bit and the Boolean functions Z1, Z2 by the letters '$F$', '$L$', '$Z$' and '$Y$', respectively.

\begin{figure}[h]
    \centering
    \includegraphics[scale=0.4]{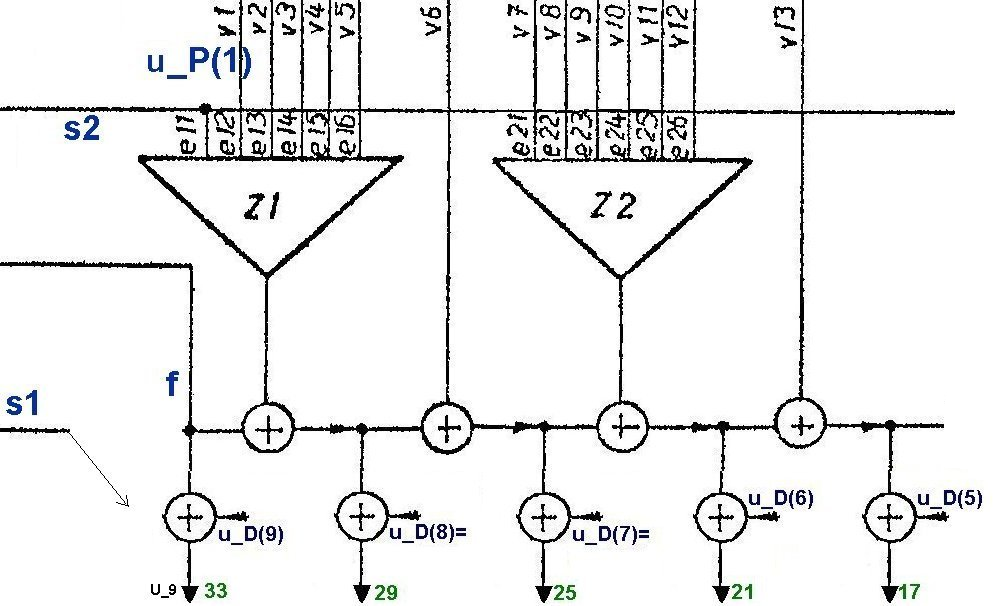}
    \caption{Area of Interest for 20 Bit Invariants}
    \label{fig:mesh1}
\end{figure}

\begin{theorem}[Toy example of invariant on 20 bits]
For the Boolean function\\
$Z(x0,x1,x2,x3,x4,x5) = x1*x2*x3*x4*x5$ and the long term key with\\
$D = 17,25,26,35,18,34,30,32,28$ and\\
$P = 27,29,31,21,33,19,26,25,22,32,23,17,24,16,18,9,5,10,35,13,36,30,34,11,2,28,14$\\
we have the invariant $\mathcal{P} = h \oplus g \oplus f \oplus e \oplus fgh \oplus egh \oplus efh \oplus efg$
for $F = 0$ and the invariant $\mathcal{R} = g \oplus f \oplus gh \oplus eh \oplus fg \oplus ef$ for $F = 1$.
\end{theorem}

\begin{proof}
We have $p = u_{21}$, $o = u_{22}$, $j = u_{27}$, $i = u_{28}$, $h = u_{29}$, $g = u_{30}$, $f = u_{31}$, $e = u_{32}$ $d = u_{33}$, $c = u_{34}$\\
We observe that $[p] \to [o]$, $[j] \to [i]$, $[h] \to [g]$, $[g] \to [f]$, $[f] \to [e]$ and $[d] \to [c]$ for one round.\\
From the combination of equations (8) and (9) the following holds:
$$e = F \oplus Z \oplus h$$
According to the theorem, we have $P(1-5) = \{27,29,31,21,33\}$, hence $Z = j*h*f*p*d$ which becomes $Z = i*f*e*o*c$ in the next round.\\
We have that $\mathcal{P} = h \oplus g \oplus f \oplus fgh \oplus e * ( 1 \oplus gh \oplus fh \oplus fg)$.\\
In the next round we have:
$$\mathcal{P} = g \oplus f \oplus e \oplus efg \oplus (F \oplus Z \oplus h) * (1 \oplus fg \oplus eg \oplus ef)$$
Then, if we set $F = 0$ and replace $Z$ by the expression $i*f*e*o*c$ we get back
$$\mathcal{P} = h \oplus g \oplus f \oplus e \oplus fgh \oplus egh \oplus efh \oplus efg$$
as expected.\\
Similarly, we have that $\mathcal{R} = g \oplus f \oplus gh \oplus fg \oplus e*(h \oplus f)$ which results in 
$$\mathcal{R} = f \oplus e \oplus fg \oplus oef \oplus (F \oplus Z \oplus h) * (g \oplus e)$$
in the next round and if we set $F = 1$ and replace $Z$ by the expression $i*f*e*o*c$ we get back
$$\mathcal{R} = g \oplus f \oplus gh \oplus eh \oplus fg \oplus ef$$
as expected.
\end{proof}

\newpage
\begin{theorem}[Invariant for $D(8)=32$]
For the Boolean function of T-310, if $F = 0$, $L = 0$ and for the long term key with
$D = 0,12,16,4,36,28,20,32,24$ and\\
$P = 22,29,18,31,30,32,35,27,34,28,33,26,20,24,21,17,13,25,27,8,19,36,23,16,4,15,14$\\
we have the invariant\\
$\mathcal{P} = h \oplus g \oplus f \oplus e \oplus gh \oplus fh \oplus eh \oplus fg \oplus eg \oplus ef \oplus fgh \oplus egh \oplus efh \oplus efg \oplus efgh$.
\end{theorem}

\begin{proof}
We observe that $[h] \to [g]$, $[g] \to [f]$, $[f] \to [e]$ for one round.\\
According to the theorem, $D(8) = e$ and from the combination of equations (8) and (9) we have:
$$e = F \oplus Z \oplus h$$
Hence, in the next round $\mathcal{P}$ becomes:\\
$\mathcal{P} = g \oplus f \oplus e \oplus fg \oplus eg \oplus ef \oplus efg \oplus (F \oplus Z \oplus h) * (1 \oplus g \oplus f \oplus e \oplus fg \oplus eg \oplus ef \oplus efg)$\\
If we set $F = 0$ then we have:\\
$\mathcal{P} = h \oplus g \oplus f \oplus e \oplus gh \oplus fh \oplus eh \oplus fg \oplus eg \oplus ef \oplus fgh \oplus egh \oplus efh \oplus efg \oplus efgh \oplus Z * (1 \oplus g \oplus f \oplus e \oplus fg \oplus eg \oplus ef \oplus efg)$\\
Now we exploit the fact that for the Boolean function of T-310 when $L = 0$:\\
\begin{equation}
Z * (1 \oplus x2 \oplus x4 \oplus x5 \oplus x2*x4 \oplus x2*x5 \oplus x4*x5 \oplus x2*x4*x5) = 0
\end{equation}
According to the theorem $P(2) = x2 = h$, $P(4) = x4 = f$ and $P(5) = x5 = g$ which will become $g$, $e$ and $f$ in the next round, respectively. Hence, we have:
$$Z * (1 \oplus g \oplus f \oplus e \oplus fg \oplus eg \oplus ef \oplus efg) = 0$$
Therefore, in the next round we have:\\
$\mathcal{P} = h \oplus g \oplus f \oplus e \oplus gh \oplus fh \oplus eh \oplus fg \oplus eg \oplus ef \oplus fgh \oplus egh \oplus efh \oplus efg \oplus efgh$\\
and $\mathcal{P}$ is an invariant for the above key, if $F = 0$ and $L = 0$.
\end{proof}

\begin{remark}
The above theorem will also be true whenever $\{29,30,31\} \subset \{P(2),P(4),P(5)\}$
\end{remark}
\pagebreak
Now we will use equation (5.1) again in order to expand Theorem 5.3.2 to hold for both $L = 0$ and $L = 1$.

\begin{theorem}[Invariant for both $L=0$ and $L=1$]
For the Boolean function of T-310 and the long term key with $D = 0,4,8,12,24,20,28,36,32$ and\\
$P = 18,29,26,31,30,36,35,19,21,23,24,34,33,15,32,14,12,3,20,16,4,2,13,11,1,12,10$\\
if $F = 1$, we have the invariant $\mathcal{P} = abcdijkl * (1 \oplus h \oplus g \oplus f \oplus e \oplus gh \oplus fh \oplus eh \oplus fg \oplus eg \oplus ef \oplus fgh \oplus egh \oplus efh \oplus efg \oplus efgh)$.
\end{theorem}

\begin{proof}
We observe that $D(7) = i$, $D(8) = P(6) = a$, $D(9) = e$. Thus, from equations (7), (9) and the combination of equations (8) and (9) we have:
$$i = l \oplus h$$
$$e = F \oplus d$$
$$a = F \oplus Z \oplus h$$
If we let $x = abcdijkl$, then in the next round $x$ will become:
$$x = abcijk*(l \oplus h)*(F \oplus Z \oplus h)$$
and if we let $y = 1 \oplus h \oplus g \oplus f \oplus e \oplus gh \oplus fh \oplus eh \oplus fg \oplus eg \oplus ef \oplus fgh \oplus egh \oplus efh \oplus efg \oplus efgh$, then in the next round $y$ will become:
$$y = 1 \oplus g \oplus f \oplus e \oplus fg \oplus eg \oplus ef \oplus efg \oplus (F \oplus d) * (1 \oplus g \oplus f \oplus e \oplus fg \oplus eg \oplus ef \oplus efg)$$
If we set $F=1$, then $x$ and $y$ become:
$$x = abcijkl * (1 \oplus h) \oplus Zabcijk*(l \oplus h)$$
$$y = d * (1 \oplus g \oplus f \oplus e \oplus fg \oplus eg \oplus ef \oplus efg)$$
Clearly, $\mathcal{P} = x*y$. So in the next round $\mathcal{P}$ becomes:\\
$\mathcal{P} = abcdijkl * (1 \oplus h)*(1 \oplus g \oplus f \oplus e \oplus fg \oplus eg \oplus ef \oplus efg) \oplus Zabcdijk*(l \oplus h)*(1 \oplus g \oplus f \oplus e \oplus fg \oplus eg \oplus ef \oplus efg)$\\
Now, we will exploit two facts. As mentioned in Theorem 5.3.2, when $L=0$:
$$Z * (1 \oplus x2 \oplus x4 \oplus x5 \oplus x2*x4 \oplus x2*x5 \oplus x4*x5 \oplus x2*x4*x5) = 0$$
Furthermore, when $L=1$:
$$Z * (1 \oplus x2 \oplus x4 \oplus x5 \oplus x2*x4 \oplus x2*x5 \oplus x4*x5 \oplus x2*x4*x5)*x3 = 0$$
According to the theorem $P(2) = x2 = h$, $P(3) = x3 = k$, $P(4) = x4 = f$ and $P(5) = x5 = g$ which will become $g$, $j$, $e$ and $f$ in the next round, respectively.\\
Thus, if either $L=0$ or $L=1$, we have:
$$Zabcdijk*(l \oplus h)*(1 \oplus g \oplus f \oplus e \oplus fg \oplus eg \oplus ef \oplus efg) = 0$$
Hence, in either case $\mathcal{P}$ is simplified into 
$$\mathcal{P} = abcdijkl * (1 \oplus h \oplus g \oplus f \oplus e \oplus gh \oplus fh \oplus eh \oplus fg \oplus eg \oplus ef \oplus fgh \oplus egh \oplus efh \oplus efg \oplus efgh)$$
and, thus, it is an invariant for both $L=0$ and $L=1$.
\end{proof}
\begin{remark}
The above theorem is also true whenever $P(3) \in \{25,26,27,32,33,34,35\}$
\end{remark}


\chapter{Conclusion and Further Work}

\label{ch:conclusions}

T-310 contains a block cipher in a stream cipher mode and only 10 bits are used from a total of 1651 rounds to just encrypt a single character. Approximately 3$\%$ of long term KT1 keys are weak against Linear Cryptanalysis (see Section 21.20 in [9]). In our thesis, we have explained how a KT1 key can exhibit one bit correlations or other linear invariances for more than one bit, and demonstrated how these invariances can be exploited. From our observations, we concluded that it is not possible to extend the attack in section 4.3.2 in order to break the whole cipher and we moved on to Generalised Linear Cryptanalysis. Based on our results, we realised that a long term key can exhibit a large number of non linear invariances which hold for one round and thus help us recover a larger proportion of the secret key bits. We have also combined a property of Boolean functions (see Theorem 3.0.1) to generate non linear invariances.\\\\
We believe that in our ongoing research we will be able to find KT1 keys which will be weak against Generalised Linear Cryptanalysis. Based on our findings, we are also confident that in the near future attacks can be developed to expoit long term keys which are weak against Generalised Linear Cryptanalysis, and it is expected that these attacks will have more devastating effects on the security of T-310.

\addcontentsline{toc}{chapter}{Bibliography}
\bibliographystyle{alpha}
\bibliography{bibliography/bibliography}

\appendix
\chapter{Description of KT1 Keys}
Below is a description of the KT1 class of long term keys found in [94].\\
$( P , D , \alpha ) \in K T 1 \Leftrightarrow \text { all of the following hold: }$
\begin{spacing}{1.3}
\begin{itemize}[noitemsep,topsep=0pt]
\item $D$ and $P$ are injective
\item $P ( 3 ) = 33 , P ( 7 ) = 5 , P ( 9 ) = 9 , P ( 15 ) = 21 , P ( 18 ) = 25 , P ( 24 ) = 29$
\item $\text { Let } W = \{ 5,9,21,25,29,33 \}$
\item $\forall _ { 1 \geq i \geq 9} \text{ }D ( i ) \notin W$
\item $\alpha \notin W \quad \text { (note: cf. also Fig. } 9.16 \text { page } 30 )$
\item $\text { Let } T = ( \{ 0,1 , \ldots , 12 \} \backslash W ) \cap ( \{ P ( 1 ) , P ( 2 ) , \ldots , P ( 24 ) \} \cup \{ D ( 4 ) , D ( 5 ) , \ldots , D ( 9 ) \} \cup \{ \alpha \})$
\item $\text { Let } U = ( \{ 13 , \dots , 36 \} \backslash W ) \cap ( \{ P ( 26 ) , P ( 27 ) \} \cup \{ D ( 1 ) , D ( 2 ) , D ( 3 ) \} )$
\item $| T \backslash \{ P ( 25 ) \} | \text{ }+\text{ } | U \backslash \{ P ( 25 ) \} | \leq 12$
\item $D ( 1 ) = 0$
\item $\text { There exist } \left\{ j _ { 1 } , j _ { 2 } , \dots , j _ { 7 } , j _ { 8 } \right\} \text { a permutation of } \{ 2,3 , \dots , 9 \} \text { which defines }$\\
$D ( i ) \text { for every } i \in \{ 2,3 , \dots , 9 \} \text { as follows: }$\\
$D \left( j _ { 1 } \right) = 4 , D \left( j _ { 2 } \right) = 4 j _ { 1 } , D \left( j _ { 3 } \right) = 4 j _ { 2 } , \dots , D \left( j _ { 8 } \right) = 4 j _ { 7 }$
\item $P ( 20 ) = 4 j _ { 8 } \text { (note: this value is not any of the } D ( i ) )$
\item $( D ( 5 ) , D ( 6 ) ) \in \{ 8,12,16 \} \times \{ 20,28,32 \} \cup \{ 24,28,32 \} \times \{ 8,12,16 \}$
\item $P ( 6 ) = D ( 8 ) , P ( 13 ) = D ( 7 )$
\item $P ( 27 ) \neq 0 \bmod 4$
\item $\forall _ { 1 \geq l \geq 9 } \exists _ { 1 \geq i \geq 26 } P ( i ) = 4 \cdot l$
\item $D ( 3 ) \in \{ P ( 1 ) , P ( 2 ) , P ( 4 ) , P ( 5 ) \}$
\item $D ( 4 ) \quad \notin \quad \{ P ( 14 ) , P ( 16 ) , P ( 17 ) , P ( 19 ) \}$
\item $\{ P ( 8 ) , P ( 10 ) , P ( 11 ) , P ( 12 ) \} \cap \{ D ( 4 ) , D ( 5 ) , D ( 6 ) \} = \emptyset$
\end{itemize}
\end{spacing}

\chapter{Description of KT2 Keys}
Below is a description of the KT1 class of long term keys found in [94].\\
$( P , D , \alpha ) \in K T 2 \Leftrightarrow \text { all of the following hold: }$
\begin{spacing}{1.5}
\begin{itemize}[noitemsep,topsep=0pt]
\item $D \text { and } P \text { are injective }$
\item $P ( 3 ) = 33 , P ( 7 ) = 5 , P ( 9 ) = 9 , P ( 15 ) = 21 , P ( 18 ) = 25 , P ( 24 ) = 29$
\item $\text { Let } W = \{ 5,9,21,25,29,33 \}$
\item $\forall _ { 1 \geq i \geq 9 }\text{ } D ( i ) \notin W$
\item $\alpha \notin W$
\item $\text { Let } T = ( \{ 0,1 , \ldots , 12 \} \backslash W ) \cap ( \{ P ( 1 ) , P ( 2 ) , \ldots , P ( 24 ) \} \cup \{ D ( 4 ) , D ( 5 ) , \ldots , D ( 9 ) \} \cup \{ \alpha \})$
\item $\text { Let } U = ( \{ 13 , \ldots , 36 \} \backslash W ) \cap ( \{ P ( 26 ) , P ( 27 ) \} \cup \{ D ( 1 ) , D ( 2 ) , D ( 3 ) \} )$
\item $| T \backslash \{ P ( 25 ) \} | + | U \backslash \{ P ( 25 ) \} | \text{ }\leq 12$
\item $A = \{ D ( 1 ) , D ( 2 ) , D ( 3 ) , D ( 4 ) , D ( 5 ) , D ( 6 ) , D ( 7 ) , D ( 8 ) , D ( 9 ) \} \cup \{ P ( 6 ) , P ( 13 ) , P ( 20 ) , P ( 27 ) \}$
\item $A _ { 1 } = \{ D ( 1 ) , D ( 2 ) \} \cup \{ P ( 27 ) \}$
\item $A _ { 2 } = \{ D ( 3 ) , D ( 4 ) \} \cup \{ P ( 20 ) \}$
\item $A _ { 3 } = \{ D ( 5 ) , D ( 6 ) \} \cup \{ P ( 13 ) \}$
\item $A _ { 4 } = \{ D ( 7 ) , D ( 8 ) \} \cup \{ P ( 6 ) \}$
\item $\forall ( i , j ) \in \{ 1 , \ldots , 27 \} \times \{ 1 , \dots , 9 \} : P _ { i } \neq D _ { j }$
\item $\exists j _ { 1 } \in \{ 1 , \dots , 7 \} : D _ { j _ { 1 } } = 0$
\item $\{ D ( 8 ) , D ( 9 ) \} \subset \{ 4,8 , \dots , 36 \} \subset A$
\item $\forall ( i , j ) \in \overline { 1,27 } \times \overline { 1,9 } : P _ { i } \neq D _ { j }$
\item $\exists j _ { 1 } \in \overline { 1,7 } : D _ { j _ { 1 } } = 0$
\item $\left\{ D _ { 8 } , D _ { 9 } \right\} \subset \{ 4,8 , \dots , 36 \} \subset A$
\item $\exists \left( j _ { 2 } , j _ { 3 } \right) \in ( \{ j \in \overline { 1,4 } | D _ { j _ { 1 } } \notin A _ { j } \} ) ^ { 2 } \wedge$
\item $\exists \left( j _ { 4 } , j _ { 5 } \right) \in \left( \overline { 1,4 } \backslash \left\{ j _ { 1 } , 2 j _ { 2 } - 1,2 j _ { 2 } \right\} \right) \times \left( \overline { 5,8 } \backslash \left\{ j _ { 1 } , 2 j _ { 2 } - 1,2 j _ { 2 } \right\} \right) \wedge$
\item $\exists j _ { 6 } \in \overline { 1,9 } \backslash \left\{ j _ { 1 } , 2 j _ { 2 } - 1,2 j _ { 2 } , j _ { 4 } , j _ { 5 } \right\}$
\item $j _ { 2 } \neq j _ { 3 } \wedge \left\{ 4 j _ { 4 } , 4 j _ { 5 } \right\} \subset A _ { j _ { 2 } } \wedge$
\item $A _ { j _ { 2 } } \cap \left( \overline { 4 j _ { 1 } - 3,4 j _ { 1 } } \cup \overline { 4 j _ { 6 } - 3,4 j _ { 6 } } \right) \neq \emptyset \wedge$
\item $\left\{ 8 j _ { 2 } - 5,8 j _ { 2 } \right\} \subset A _ { j _ { 3 } } \wedge A _ { j _ { 3 } } \cap \left( \overline { 4 j _ { 1 } - 3,4 j _ { 1 } } \cup \overline { 4 j _ { 6 } - 3,4 j _ { 6 } } \right) \neq \emptyset$
\item $\{ D ( 9 ) \} \backslash ( \overline { 33,36 } \cup \{ 0 \} ) \neq \emptyset$
\item $\{ D ( 8 ) , D ( 9 ) , P ( 1 ) , P ( 2 ) , \dots , P ( 5 ) \} \backslash ( \overline { 29,32 } \cup \{ 0 \} ) \neq \emptyset$
\item $\{ D ( 7 ) , D ( 8 ) , P ( 1 ) , P ( 2 ) , \dots , P ( 6 ) \} \backslash ( \overline { 25,32 } \cup \{ 0 \} ) \neq \emptyset$
\item $\{ D ( 7 ) , D ( 9 ) , P ( 1 ) , P ( 2 ) , \ldots , P ( 6 ) \} \backslash ( \overline { 25,28 } \cup \overline { 33,36 } \cup \{ 0 \} ) \neq \emptyset$
\item $\{ D ( 6 ) , D ( 7 ) , D ( 8 ) , D ( 9 ) , P ( 1 ) , P ( 2 ) , \dots , P ( 12 ) \} \backslash ( \overline { 21,36 } \cup \{ 0 \} ) \neq \emptyset$
\item $\{ D ( 5 ) , D ( 7 ) , D ( 8 ) , D ( 9 ) , P ( 1 ) , P ( 2 ) , \dots , P ( 13 ) \} \backslash ( \overline { 17,20 } \cup \overline { 25,36 } \cup \{ 0 \} ) \neq \emptyset$
\item $\{ D ( 7 ) , D ( 8 ) , D ( 9 ) , P ( 1 ) , P ( 2 ) , \dots , P ( 6 ) \} \backslash ( \overline { 25,36 } \cup \{ 0 \} ) \neq \emptyset$
\item $\{ D ( 5 ) , D ( 6 ) , D ( 8 ) , D ( 9 ) , P ( 1 ) , P ( 2 ) , \dots , P ( 13 ) \} \backslash ( \overline { 17,24 } \cup \overline { 29,36 } \cup \{ 0 \} ) \neq \emptyset$
\item $\{ D ( 5 ) , D ( 6 ) , D ( 7 ) , D ( 9 ) , P ( 1 ) , P ( 2 ) , \dots , P ( 13 ) \} \backslash ( \overline { 17,28 } \cup \overline { 33,36 } \cup \{ 0 \} ) \neq \emptyset$
\item $\{ D ( 5 ) , D ( 6 ) , D ( 7 ) , D ( 8 ) , P ( 1 ) , P ( 2 ) , \dots , P ( 13 ) \} \backslash ( \overline { 17,32 } \cup \{ 0 \} ) \neq \emptyset$
\item $\{ D ( 5 ) , D ( 6 ) , D ( 7 ) , D ( 8 ) , D ( 9 ) , P ( 1 ) , P ( 2 ) , \dots , P ( 13 ) \} \backslash ( \overline { 17,36 } \cup \{ 0 \} ) \neq \emptyset$
\item $\{ D ( 4 ) , D ( 5 ) , \ldots , D ( 9 ) , P ( 1 ) , P ( 2 ) , \ldots , P ( 19 ) \} \backslash ( \overline { 13,36 } \cup \{ 0 \} ) \neq \emptyset$
\item $\{ D ( 3 ) , D ( 4 ) , \ldots , D ( 9 ) , P ( 1 ) , P ( 2 ) , \ldots , P ( 20 ) \} \backslash ( \overline { 9,36 } \cup \{ 0 \} ) \neq \emptyset$
\item the "Matrix rank $= 9$ condition" which is defined as:\\
$\text { The concrete values } D ( i ) / P ( j ) \text { inside the formulas } \mathbf { D } \left( s 1 , u _ { I ^ { 1 } } \right) \oplus \mathbf { T } \left( f , s 2 , \mathbf { P } \left( u _ { I ^ { 1 - 4 } } \right) \right)$\\
$\text { which define the } 9 \text { "fresh" outputs } I ^ { 4 } = \{ 1,5 , \dots , 33 \} \text { of } \phi \text { appear at such places }$\\
$\text { that all the } 9 \text { "fresh" outputs } I ^ { 4 } \text { of } \phi \text { are sums of non-linear parts of type } Z ( . ),$\\
$\text { plus affine parts which involve various variables in } u _ { I ^ { 2 - 4 } } , \text { plus an invertible linear }$\\
$\text { transformation } B \text { of rank } 9 \text { with the remaining } 9 \text { inputs of } I ^ { 1 } = \{ 4,8 , \dots , 36 \}.$
\end{itemize}
\end{spacing}

\chapter{One-Bit Correlations for 7 Rounds}
Using a similar method as in Theorem \ref{Preconditions6rAlphtoalph} we produced the pre-conditions needed for one-bit correlations for 7 rounds with Hamming Weight=3,4 in the fourth round for any KT1 key.

\begin{spacing}{0.1}
\textbf{Case 1a: }\bm{$[29]\to[29]$} \textbf{using the boolean functions Z2 and Z1}\\
Pre-conditions:
\begin{multicols}{2}
\begin{itemize}[noitemsep,topsep=0pt]
\item{$D(6)=32$}
\item{$\{23,31\}\in\{P(1),P(2),P(4),P(5)\}$}
\item{$D(8)\in\{4,8,12,16,20\}$}
\item{$D(8)-3\in\{P(7)-P(12)\}$}
\end{itemize}
\end{multicols}
\end{spacing}

\begin{table}[h]
\small
\centering
\label{my-label1}
\begin{tabular}{|c|c|c|}
\hline
$D(8)$ & internal set of 3 bits after 4 rounds & proportion of weak keys \\ \hline
4      & $[2,21,29]$    & $2^{-14.46}$                       \\ \hline
8      & $[6,21,29]$     & $2^{-11.76}$                      \\ \hline
12     & $[10,21,29]$       & $2^{-11.76}$                   \\ \hline
16     & $[14,21,29]$     & $2^{-14.46}$                     \\ \hline
20     & $[18,21,29]$        & $2^{-14.46}$                  \\ \hline
\end{tabular}
\caption{Enumeration of All Cases for $D(8)$}
\end{table}

\begin{spacing}{0.1}
\textbf{Case 1b: }\bm{$[29]\to[29]$} \textbf{using the boolean functions Z2 and Z1}\\
Pre-conditions: 
\begin{multicols}{2}
\begin{itemize}[noitemsep,topsep=0pt]
\item{$D(5)=32$}
\item{$\{19,27\}\in\{P(1),P(2),P(4),P(5)\}$}
\item{$D(8)\in\{4,8,12,16,20\}$}
\item{$D(8)-3\in\{P(7)-P(12)\}$}
\end{itemize}
\end{multicols}
\end{spacing}

\begin{table}[h]
\small
\centering
\label{my-label2}
\begin{tabular}{|c|c|c|}
\hline
$D(8)$ & internal set of 3 bits after 4 rounds & proportion of weak keys \\ \hline
4      & $[2,17,25]$    & $2^{-14.46}$                       \\ \hline
8      & $[6,17,25]$     & $2^{-11.76}$                       \\ \hline
12     & $[10,17,25]$     & $2^{-11.76}$                      \\ \hline
16     & $[14,17,25]$    & $2^{-14.46}$                      \\ \hline
20     & $[17,18,25]$   & $2^{-14.46}$                       \\ \hline
\end{tabular}
\caption{Enumeration of All Cases for $D(8)$}
\end{table}

\pagebreak
\begin{spacing}{0.1}
\textbf{Case 2: }\bm{$[29]\to[29]$} \textbf{using the boolean functions Z3 and Z1}\\
Pre-conditions: 
\begin{multicols}{2}
\begin{itemize}[noitemsep,topsep=0pt]
\item{$D(4)=32$}
\item{$D(5)\in \{P(14),P(16),P(17),P(19)\}$}
\item{$\{15,19\}\in\{P(1),P(2),P(4),P(5)\}$}
\item{$D(8)\in\{4,16,20,24,28\}$}
\item{$D(8)-3\in\{P(14)-P(19)\}$}
\end{itemize}
\end{multicols}
\end{spacing}

\begin{table}[h]
\small
\centering
\label{my-label3}
\begin{tabular}{|c|c|c|}
\hline
$D(8)$ & internal set of 3 bits after 4 rounds & proportion of weak keys \\ \hline
4      & $[2,13,17]$      & $2^{-16.94}$                     \\ \hline
16     & $[13,14,17]$     & $2^{-16.94}$                     \\ \hline
20     & $[13,17,18]$     & $2^{-16.94}$                     \\ \hline
24     & $[13,17,22]$     & $2^{-13.88}$                     \\ \hline
28     & $[13,17,26]$    & $2^{-13.88}$                       \\ \hline
\end{tabular}
\caption{Enumeration of All Cases for $D(8)$}
\end{table}

\begin{spacing}{0.1}
\textbf{Case 3: }\bm{$[29]\to[29]$} \textbf{using the boolean functions Z4 and Z1}\\
Pre-conditions: 
\begin{multicols}{2}
\begin{itemize}[noitemsep,topsep=0pt]
\item{$D(3)=32$}
\item{$D(2)\in \{P(21),P(22),P(23),P(25),P(26)\}$}
\item{$\{7,11\}\in\{P(1),P(2),P(4),P(5)\}$}
\item{$D(8)\in\{4,16,20\}$}
\item{$D(8)-3\in\{P(21),P(22),P(23),P(25),P(26)\}$}
\end{itemize}
\end{multicols}
\end{spacing}

\begin{table}[h]
\small
\centering
\label{my-label4}
\begin{tabular}{|c|c|c|}
\hline
$D(8)$ & internal set of 3 bits after 4 rounds & proportion of weak keys \\ \hline
4      & $[2,5,9]$     & $2^{-16.58}$                        \\ \hline
16     & $[5,9,14]$    & $2^{-16.58}$                     \\ \hline
20     & $[5,9,18]$     & $2^{-16.58}$                       \\ \hline
\end{tabular}
\caption{Enumeration of All Cases for $D(8)$}
\end{table}

\begin{spacing}{0.1}
\textbf{Case 4a: }\bm{$[25]\to[25]$} \textbf{using the boolean functions Z2 and Z1}\\
Pre-conditions: 
\begin{multicols}{2}
\begin{itemize}[noitemsep,topsep=0pt]
\item{$D(6)=28$}
\item{$\{23,31\}\in\{P(1),P(2),P(4),P(5)\}$}
\item{$\{D(7),D(8)\}\in\{4,8,12,16,20\}$}
\item{$\{D(7)-3,D(8)-3\}\in\{P(7)-P(12)\}$}
\end{itemize}
\end{multicols}
\end{spacing}

\begin{table}[h]
\small
\centering
\label{my-label5}
\begin{tabular}{|c|c|c|c|c|}
\cline{1-1} \cline{3-5}
$D(7)$ &                   & $D(8)$ & internal set of 4 bits after 4 rounds & proportion of weak keys \\ \hline
4      & $\leftrightarrow$ & 8      & $[2,6,21,29]$     & $2^{-16.01}$                    \\ \hline
4      & $\leftrightarrow$ & 12     & $[2,10,21,29]$           & $2^{-16.01}$             \\ \hline
4      & $\leftrightarrow$ & 16     & $[2,14,21,29]$        & $2^{-19.01}$                \\ \hline
4      & $\leftrightarrow$ & 20     & $[2,18,21,29]$        & $2^{-19.01}$                 \\ \hline
8      & $\leftrightarrow$ & 12     & $[6,10,21,29]$          & $2^{-13.36}$               \\ \hline
8      & $\leftrightarrow$ & 16     & $[6,14,21,29]$       & $2^{-16.01}$                 \\ \hline
8      & $\leftrightarrow$ & 20     & $[6,18,21,29]$        & $2^{-16.01}$                \\ \hline
12     & $\leftrightarrow$ & 16     & $[10,14,21,29]$      & $2^{-16.01}$                 \\ \hline
12     & $\leftrightarrow$ & 20     & $[10,18,21,29]$        & $2^{-16.01}$               \\ \hline
16     & $\leftrightarrow$ & 20     & $[16,18,21,29]$       & $2^{-19.01}$                 \\ \hline
\end{tabular}
\caption{Enumeration of All Cases for $D(7)$ and $D(8)$}
\end{table}
\pagebreak

\begin{spacing}{0.1}
\textbf{Case 4b: }\bm{$[25]\to[25]$} \textbf{using the boolean functions Z2 and Z1}\\
Pre-conditions: 
\begin{multicols}{2}
\begin{itemize}[noitemsep,topsep=0pt]
\item{$D(5)=28$}
\item{$\{19,27\}\in\{P(1),P(2),P(4),P(5)\}$}
\item{$\{D(7),D(8)\}\in\{4,8,12,16,20\}$}
\item{$\{D(7)-3,D(8)-3\}\in\{P(7)-P(12)\}$}
\end{itemize}
\end{multicols}
\end{spacing}

\begin{table}[h]
\small
\centering
\label{my-label6}
\begin{tabular}{|c|c|c|c|c|}
\cline{1-1} \cline{3-5}
$D(7)$ &                   & $D(8)$ & internal set of 4 bits after 4 rounds & proportion of weak keys \\ \hline
4      & $\leftrightarrow$ & 8      & $[2,6,17,25]$        & $2^{-16.01}$                 \\ \hline
4      & $\leftrightarrow$ & 12     & $[2,10,17,25]$        & $2^{-16.01}$                \\ \hline
4      & $\leftrightarrow$ & 16     & $[2,14,17,25]$        & $2^{-19.01}$                \\ \hline
4      & $\leftrightarrow$ & 20     & $[2,17,18,25]$          & $2^{-19.01}$              \\ \hline
8      & $\leftrightarrow$ & 12     & $[6,10,17,25]$            & $2^{-13.36}$             \\ \hline
8      & $\leftrightarrow$ & 16     & $[6,14,17,25]$         & $2^{-16.01}$               \\ \hline
8      & $\leftrightarrow$ & 20     & $[6,17,18,25]$         & $2^{-16.01}$               \\ \hline
12     & $\leftrightarrow$ & 16     & $[10,14,17,25]$       & $2^{-16.01}$                \\ \hline
12     & $\leftrightarrow$ & 20     & $[10,17,18,25]$       & $2^{-16.01}$                \\ \hline
16     & $\leftrightarrow$ & 20     & $[14,17,18,25]$         & $2^{-19.01}$              \\ \hline
\end{tabular}
\caption{Enumeration of All Cases for $D(7)$ and $D(8)$}
\end{table}

\begin{spacing}{0.1}
\textbf{Case 5: }\bm{$[25]\to[25]$} \textbf{using the boolean functions Z3 and Z1}\\
Pre-conditions: 
\begin{multicols}{2}
\begin{itemize}[noitemsep,topsep=0pt]
\item{$D(4)=28$}
\item{$D(5)\in \{P(14),P(16),P(17),P(19)\}$}
\item{$\{15,19\}\in\{P(1),P(2),P(4),P(5)\}$}
\item{$\{D(7),D(8)\}\in\{4,16,20,24\}$}
\item{$\{D(7)-3,D(8)-3\}\in\{P(14)-P(19)\}$}
\end{itemize}
\end{multicols}
\end{spacing}

\begin{table}[h]
\small
\centering
\label{my-label7}
\begin{tabular}{|c|c|c|c|c|}
\cline{1-1} \cline{3-5}
$D(7)$ &                   & $D(8)$ & internal set of 4 bits after 4 rounds & proportion of weak keys \\ \hline
4      & $\leftrightarrow$ & 16     & $[2,13,14,17]$      & $2^{-22.00}$                  \\ \hline
4      & $\leftrightarrow$ & 20     & $[2,13,17,18]$     & $2^{-22.00}$                   \\ \hline
4      & $\leftrightarrow$ & 24     & $[2,13,17,22]$       & $2^{-18.48}$                 \\ \hline
16     & $\leftrightarrow$ & 20     & $[13,14,17,18]$        & $2^{-22.00}$               \\ \hline
16     & $\leftrightarrow$ & 24     & $[13,14,17,22]$         & $2^{-18.48}$               \\ \hline
20     & $\leftrightarrow$ & 24     & $[13,17,18,22]$         & $2^{-18.48}$               \\ \hline
\end{tabular}
\caption{Enumeration of All Cases for $D(7)$ and $D(8)$}
\end{table}

\begin{spacing}{0.1}
\textbf{Case 6: }\bm{$[25]\to[25]$} \textbf{using the boolean functions Z4 and Z1}\\
Pre-conditions: 
\begin{multicols}{2}
\begin{itemize}[noitemsep,topsep=0pt]
\item{$D(3)=28$}
\item{$D(2)\in \{P(21),P(22),P(23),P(25),P(26)\}$}
\item{$\{7,11\}\in\{P(1),P(2),P(4),P(5)\}$}
\item{$\{D(7),D(8)\}\in\{4,16,20,32\}$}
\item{$\{D(7)-3,D(8)-3\}\in\{P(21)-P(26)\}$}
\end{itemize}
\end{multicols}
\end{spacing}

\begin{table}[h]
\small
\centering
\label{my-label8}
\begin{tabular}{|c|c|c|c|c|}
\cline{1-1} \cline{3-5}
$D(7)$ &                   & $D(8)$ & internal set of 4 bits after 4 rounds & proportion of weak keys \\ \hline
4      & $\leftrightarrow$ & 16     & $[2,5,9,14]$         & $2^{-21.07}$                   \\ \hline
4      & $\leftrightarrow$ & 20     & $[2,5,9,18]$           & $2^{-21.07}$               \\ \hline
32      &  & 4     & $[2,5,9,30]$            & $2^{-19.34}$              \\ \hline
16     & $\leftrightarrow$ & 20     & $[5,9,14,18]$    & $2^{-21.07}$                     \\ \hline
32     &  & 16     & $[5,9,14,30]$      & $2^{-19.34}$                   \\ \hline
32     &  & 20     & $[5,9,18,30]$       & $2^{-19.34}$                  \\ \hline
\end{tabular}
\caption{Enumeration of All Cases for $D(7)$ and $D(8)$}
\end{table}

\pagebreak

\begin{spacing}{0.1}
\textbf{Case 7: }\bm{$[21]\to[21]$} \textbf{using the boolean functions Z2 and Z1/Z2}\\
Pre-conditions: 
\begin{multicols}{2}
\begin{itemize}[noitemsep,topsep=0pt]
\item{$D(5)=24$}
\item{$\{19,27\}\in\{P(1)-P(12)\}$}
\item{$D(6)\in\{8,12,16\}$}
\item{$D(8)\in\{4,8,12,16,20\}$}
\item{$\{D(6)-3,D(8)-3\}\in\{P(7)-P(12)\}$}
\end{itemize}
\end{multicols}
\end{spacing}

\begin{table}[h]
\small
\centering
\label{my-label9}
\begin{tabular}{|c|c|c|c|c|}
\cline{1-1} \cline{3-5}
$D(8)$ &                   & $D(6)$ & internal set of 4 bits after 4 rounds & proportion of weak keys \\ \hline
4      &                   & 8      & $[2,6,17,25]$         & $2^{-15.63}$                 \\ \hline
4      &                   & 12     & $[2,10,17,25]$         & $2^{-15.63}$               \\ \hline
4      &                   & 16     & $[2,14,17,25]$        & $2^{-19.12}$                \\ \hline
8      & $\leftrightarrow$ & 12     & $[6,10,17,25]$            & $2^{-11.55}$            \\ \hline
8      & $\leftrightarrow$ & 16     & $[6,14,17,25]$             & $2^{-14.61}$           \\ \hline
20     &                   & 8      & $[6,17,18,25]$            & $2^{-14.61}$             \\ \hline
12     & $\leftrightarrow$ & 16     & $[10,14,17,25]$      & $2^{-15.63}$                 \\ \hline
20     &                   & 12     & $[10,17,18,25]$         & $2^{-19.12}$              \\ \hline
20     &                   & 16     & $[14,17,18,25]$        & $2^{-15.63}$               \\ \hline
\end{tabular}
\caption{Enumeration of All Cases for $D(6)$ and $D(8)$}
\end{table}

\begin{spacing}{0.1}
\textbf{Case 8: }\bm{$[21]\to[21]$} \textbf{using the boolean functions Z3 and Z1/Z2}\\
Pre-conditions: 
\begin{multicols}{2}
\begin{itemize}[noitemsep,topsep=0pt]
\item{$D(4)=24$}
\item{$D(5)\in \{P(14),P(16),P(17),P(19)\}$}
\item{$\{15,19\}\in\{P(1)-P(12)\}$}
\item{$D(6)\in\{16,20,28\}$}
\item{$D(8)\in\{4,16,20,28\}$}
\item{$\{D(6)-3,D(8)-3\}\in\{P(14)-P(19)\}$}
\end{itemize}
\end{multicols}
\end{spacing}

\begin{table}[h]
\small
\centering
\label{my-label10}
\begin{tabular}{|c|c|c|c|c|}
\cline{1-1} \cline{3-5}
$D(8)$ &                   & $D(6)$ & internal set of 4 bits after 4 rounds & proportion of weak keys \\ \hline
4      &                   & 16     & $[2,13,14,17]$       & $2^{-19.21}$                  \\ \hline
4      &                   & 20     & $[2,13,17,18]$        & $2^{-19.21}$                \\ \hline
4      &                   & 28     & $[2,13,17,26]$         & $2^{-15.63}$               \\ \hline
16     & $\leftrightarrow$ & 20     & $[13,14,17,18]$          & $2^{-18.20}$             \\ \hline
16     & $\leftrightarrow$ & 28     & $[13,14,17,26]$             & $2^{-14.61}$          \\ \hline
20     & $\leftrightarrow$ & 28     & $[13,17,18,26]$           & $2^{-14.61}$             \\ \hline
\end{tabular}
\caption{Enumeration of All Cases for $D(6)$ and $D(8)$}
\end{table}

\begin{spacing}{0.1}
\textbf{Case 9: }\bm{$[21]\to[21]$} \textbf{using the boolean functions Z4 and Z1/Z2}\\
Pre-conditions: 
\begin{multicols}{2}
\begin{itemize}[noitemsep,topsep=0pt]
\item{$D(3)=24$}
\item{$D(2)\in \{P(21),P(22),P(23),P(25),P(26)\}$}
\item{$\{7,11\}\in\{P(1)-P(12)\}$}
\item{$D(6)\in\{16,20,32\}$}
\item{$D(8)\in\{4,16,20,32\}$}
\item{$\{D(6)-3,D(8)-3\}\in\{P(21)-P(26)\}$}
\end{itemize}
\end{multicols}
\end{spacing}

\begin{table}[h]
\small
\centering
\label{my-label11}
\begin{tabular}{|c|c|c|c|c|}
\cline{1-1} \cline{3-5}
$D(8)$ &                   & $D(6)$ & internal set of 4 bits after 4 rounds & proportion of weak keys \\ \hline
4      &                   & 16     & $[2,5,9,14]$        & $2^{-19.59}$                  \\ \hline
4      &                   & 20     & $[2,5,9,18]$           & $2^{-19.59}$               \\ \hline
4      &                   & 32     & $[2,5,9,30]$           & $2^{-16.59}$               \\ \hline
16     & $\leftrightarrow$ & 20     & $[5,9,14,18]$       & $2^{-18.58}$                  \\ \hline
16     & $\leftrightarrow$ & 32     & $[5,9,14,30]$          & $2^{-15.58}$               \\ \hline
20     & $\leftrightarrow$ & 32     & $[5,9,18,30]$          & $2^{-15.58}$               \\ \hline
\end{tabular}
\caption{Enumeration of All Cases for $D(6)$ and $D(8)$}
\end{table}
\pagebreak
\chapter{Invariants on 12 Bits}
When using the Boolean function mentioned in Theorem 5.2.3 with $L = 1$ and the long term key with $D = 1,13,3,2,11,12,32,28,36$ and \\
$P = 35,27,26,34,31,29,25,14,7,22,15,33,8,30,6,10,23,4,24,18,9,20,17,19,16,5,21$\\
we observed that the invariant space has dimension 51. We list some of the invariants below.
\begin{spacing}{0.1}
\begin{footnotesize}
\begin{enumerate}
\item $d+c+b+a$
\item $bd+ac$
\item $cd+bc+ad+ab$
\item $bcd+acd+abd+abc$
\item $abcd$
\item $j+i+h+gl+gj+fk+fi+eh+dj+di+dh+cj+ci+ch+bj+bi+bh+aj+ai+ah+dgl+dgj+dfk+dfi+deh+cgl+cgj+cfk+cfi+ceh+cdj+cdi+cdh+bgl+bgj+bfk+bfi+beh+bdj+bdi+bdh+bcj+bci+bch+agl+agj+afk+afi+aeh+adj+adi+adh+acj+aci+ach+abj+abi+abh+cdgl+cdgj+cdfk+cdfi+cdeh+bdgl+bdgj+bdfk+bdfi+bdeh+bcgl+bcgj+bcfk+bcfi+bceh+bcdj+bcdi+bcdh+adgl+adgj+adfk+adfi+adeh+acgl+acgj+acfk+acfi+aceh+acdj+acdi+acdh+abgl+abgj+abfk+abfi+abeh+abdj+abdi+abdh+abcj+abci+abch+bcdgl+bcdgj+bcdfk+bcdfi+bcdeh+acdgl+acdgj+acdfk+acdfi+acdeh+abdgl+abdgj+abdfk+abdfi+abdeh+abcgl+abcgj+abcfk+abcfi+abceh+abcdj+abcdi+abcdh+abcdgl+abcdgj+abcdfk+abcdfi+abcdeh$
\item $jl+ik+ij+hj+hi+gi+gh+fl+fh+ek+eg+djl+dik+dij+dhj+dhi+dgi+dgh+dfl+dfh+dek+deg+cjl+cik+cij+chj+chi+cgi+cgh+cfl+cfh+cek+ceg+bjl+bik+bij+bhj+bhi+bgi+bgh+bfl+bfh+bek+beg+ajl+aik+aij+ahj+ahi+agi+agh+afl+afh+aek+aeg+cdjl+cdik+cdij+cdhj+cdhi+cdgi+cdgh+cdfl+cdfh+cdek+cdeg+bdjl+bdik+bdij+bdhj+bdhi+bdgi+bdgh+bdfl+bdfh+bdek+bdeg+bcjl+bcik+bcij+bchj+bchi+bcgi+bcgh+bcfl+bcfh+bcek+bceg+adjl+adik+adij+adhj+adhi+adgi+adgh+adfl+adfh+adek+adeg+acjl+acik+acij+achj+achi+acgi+acgh+acfl+acfh+acek+aceg+abjl+abik+abij+abhj+abhi+abgi+abgh+abfl+abfh+abek+abeg+bcdjl+bcdik+bcdij+bcdhj+bcdhi+bcdgi+bcdgh+bcdfl+bcdfh+bcdek+bcdeg+acdjl+acdik+acdij+acdhj+acdhi+acdgi+acdgh+acdfl+acdfh+acdek+acdeg+abdjl+abdik+abdij+abdhj+abdhi+abdgi+abdgh+abdfl+abdfh+abdek+abdeg+abcjl+abcik+abcij+abchj+abchi+abcgi+abcgh+abcfl+abcfh+abcek+abceg+abcdjl+abcdik+abcdij+abcdhj+abcdhi+abcdgi+abcdgh+abcdfl+abcdfh+abcdek+abcdeg$
\item $l+k+j+i+h+g+f+e+kl+jl+jk+il+ik+ij+hl+hk+hj+hi+gl+gk+gj+gi+gh+fl+fk+fj+fi+fh+fg+el+ek+ej+ei+eh+eg+ef+jkl+ikl+ijl+ijk+hkl+hjl+hjk+hil+hik+hij+gkl+gjl+gjk+gil+gik+gij+ghl+ghk+ghj+ghi+fkl+fjl+fjk+fil+fik+fij+fhl+fhk+fhj+fhi+fgl+fgk+fgj+fgi+fgh+ekl+ejl+ejk+eil+eik+eij+ehl+ehk+ehj+ehi+egl+egk+egj+egi+egh+efl+efk+efj+efi+efh+efg+ijkl+hjkl+hikl+hijl+hijk+gjkl+gikl+gijl+gijk+ghkl+ghjl+ghjk+ghil+ghik+ghij+fjkl+fikl+fijl+fijk+fhkl+fhjl+fhjk+fhil+fhik+fhij+fgkl+fgjl+fgjk+fgil+fgik+fgij+fghl+fghk+fghj+fghi+ejkl+eikl+eijl+eijk+ehkl+ehjl+ehjk+ehil+ehik+ehij+egkl+egjl+egjk+egil+egik+egij+eghl+eghk+eghj+eghi+efkl+efjl+efjk+efil+efik+efij+efhl+efhk+efhj+efhi+efgl+efgk+efgj+efgi+efgh+hijkl+gijkl+ghjkl+ghikl+ghijl+ghijk+fijkl+fhjkl+fhikl+fhijl+fhijk+fgjkl+fgikl+fgijl+fgijk+fghkl+fghjl+fghjk+fghil+fghik+fghij+eijkl+ehjkl+ehikl+ehijl+ehijk+egjkl+egikl+egijl+egijk+eghkl+eghjl+eghjk+eghil+eghik+eghij+efjkl+efikl+efijl+efijk+efhkl+efhjl+efhjk+efhil+efhik+efhij+efgkl+efgjl+efgjk+efgil+efgik+efgij+efghl+efghk+efghj+efghi+ghijkl+fhijkl+fgijkl+fghjkl+fghikl+fghijl+fghijk+ehijkl+egijkl+eghjkl+eghikl+eghijl+eghijk+efijkl+efhjkl+efhikl+efhijl+efhijk+efgjkl+efgikl+efgijl+efgijk+efghkl+efghjl+efghjk+efghil+efghik+efghij+fghijkl+eghijkl+efhijkl+efgijkl+efghjkl+efghikl+efghijl+efghijk+efghijkl$
\item $hij+ghi+fil+fhl+fhj+fgh+ejk+ehk+egk+egj+egi+ijkl+hijl+hijk+gijl+ghjl+ghik+fikl+fijl+fhil+fhik+fhij+fgil+fghl+fghi+ejkl+eijk+ehjk+ehik+egjl+egij+eghk+eghj+eghi+efkl+efhk+efgl+efgh+dhij+dghi+dfil+dfhl+dfhj+dfgh+dejk+dehk+degk+degj+degi+chij+cghi+cfil+cfhl+cfhj+cfgh+cejk+cehk+cegk+cegj+cegi+bhij+bghi+bfil+bfhl+bfhj+bfgh+bejk+behk+begk+begj+begi+ahij+aghi+afil+afhl+afhj+afgh+aejk+aehk+aegk+aegj+aegi+dijkl+dhijl+dhijk+dgijl+dghjl+dghik+dfikl+dfijl+dfhil+dfhik+dfhij+dfgil+dfghl+dfghi+dejkl+deijk+dehjk+dehik+degjl+degij+deghk+deghj+deghi+defkl+defhk+defgl+defgh+cijkl+chijl+chijk+cgijl+cghjl+cghik+cfikl+cfijl+cfhil+cfhik+cfhij+cfgil+cfghl+cfghi+cejkl+ceijk+cehjk+cehik+cegjl+cegij+ceghk+ceghj+ceghi+cefkl+cefhk+cefgl+cefgh+cdhij+cdghi+cdfil+cdfhl+cdfhj+cdfgh+cdejk+cdehk+cdegk+cdegj+cdegi+bijkl+bhijl+bhijk+bgijl+bghjl+bghik+bfikl+bfijl+bfhil+bfhik+bfhij+bfgil+bfghl+bfghi+bejkl+beijk+behjk+behik+begjl+begij+beghk+beghj+beghi+befkl+befhk+befgl+befgh+bdhij+bdghi+bdfil+bdfhl+bdfhj+bdfgh+bdejk+bdehk+bdegk+bdegj+bdegi+bchij+bcghi+bcfil+bcfhl+bcfhj+bcfgh+bcejk+bcehk+bcegk+bcegj+bcegi+aijkl+ahijl+ahijk+agijl+aghjl+aghik+afikl+afijl+afhil+afhik+afhij+afgil+afghl+afghi+aejkl+aeijk+aehjk+aehik+aegjl+aegij+aeghk+aeghj+aeghi+aefkl+aefhk+aefgl+aefgh+adhij+adghi+adfil+adfhl+adfhj+adfgh+adejk+adehk+adegk+adegj+adegi+achij+acghi+acfil+acfhl+acfhj+acfgh+acejk+acehk+acegk+acegj+acegi+abhij+abghi+abfil+abfhl+abfhj+abfgh+abejk+abehk+abegk+abegj+abegi+cdijkl+cdhijl+cdhijk+cdgijl+cdghjl+cdghik+cdfikl+cdfijl+cdfhil+cdfhik+cdfhij+cdfgil+cdfghl+cdfghi+cdejkl+cdeijk+cdehjk+cdehik+cdegjl+cdegij+cdeghk+cdeghj+cdeghi+cdefkl+cdefhk+cdefgl+cdefgh+bdijkl+bdhijl+bdhijk+bdgijl+bdghjl+bdghik+bdfikl+bdfijl+bdfhil+bdfhik+bdfhij+bdfgil+bdfghl+bdfghi+bdejkl+bdeijk+bdehjk+bdehik+bdegjl+bdegij+bdeghk+bdeghj+bdeghi+bdefkl+bdefhk+bdefgl+bdefgh+bcijkl+bchijl+bchijk+bcgijl+bcghjl+bcghik+bcfikl+bcfijl+bcfhil+bcfhik+bcfhij+bcfgil+bcfghl+bcfghi+bcejkl+bceijk+bcehjk+bcehik+bcegjl+bcegij+bceghk+bceghj+bceghi+bcefkl+bcefhk+bcefgl+bcefgh+bcdhij+bcdghi+bcdfil+bcdfhl+bcdfhj+bcdfgh+bcdejk+bcdehk+bcdegk+bcdegj+bcdegi+adijkl+adhijl+adhijk+adgijl+adghjl+adghik+adfikl+adfijl+adfhil+adfhik+adfhij+adfgil+adfghl+adfghi+adejkl+adeijk+adehjk+adehik+adegjl+adegij+adeghk+adeghj+adeghi+adefkl+adefhk+adefgl+adefgh+acijkl+achijl+achijk+acgijl+acghjl+acghik+acfikl+acfijl+acfhil+acfhik+acfhij+acfgil+acfghl+acfghi+acejkl+aceijk+acehjk+acehik+acegjl+acegij+aceghk+aceghj+aceghi+acefkl+acefhk+acefgl+acefgh+acdhij+acdghi+acdfil+acdfhl+acdfhj+acdfgh+acdejk+acdehk+acdegk+acdegj+acdegi+abijkl+abhijl+abhijk+abgijl+abghjl+abghik+abfikl+abfijl+abfhil+abfhik+abfhij+abfgil+abfghl+abfghi+abejkl+abeijk+abehjk+abehik+abegjl+abegij+abeghk+abeghj+abeghi+abefkl+abefhk+abefgl+abefgh+abdhij+abdghi+abdfil+abdfhl+abdfhj+abdfgh+abdejk+abdehk+abdegk+abdegj+abdegi+abchij+abcghi+abcfil+abcfhl+abcfhj+abcfgh+abcejk+abcehk+abcegk+abcegj+abcegi+bcdijkl+bcdhijl+bcdhijk+bcdgijl+bcdghjl+bcdghik+bcdfikl+bcdfijl+bcdfhil+bcdfhik+bcdfhij+bcdfgil+bcdfghl+bcdfghi+bcdejkl+bcdeijk+bcdehjk+bcdehik+bcdegjl+bcdegij+bcdeghk+bcdeghj+bcdeghi+bcdefkl+bcdefhk+bcdefgl+bcdefgh+acdijkl+acdhijl+acdhijk+acdgijl+acdghjl+acdghik+acdfikl+acdfijl+acdfhil+acdfhik+acdfhij+acdfgil+acdfghl+acdfghi+acdejkl+acdeijk+acdehjk+acdehik+acdegjl+acdegij+acdeghk+acdeghj+acdeghi+acdefkl+acdefhk+acdefgl+acdefgh+abdijkl+abdhijl+abdhijk+abdgijl+abdghjl+abdghik+abdfikl+abdfijl+abdfhil+abdfhik+abdfhij+abdfgil+abdfghl+abdfghi+abdejkl+abdeijk+abdehjk+abdehik+abdegjl+abdegij+abdeghk+abdeghj+abdeghi+abdefkl+abdefhk+abdefgl+abdefgh+abcijkl+abchijl+abchijk+abcgijl+abcghjl+abcghik+abcfikl+abcfijl+abcfhil+abcfhik+abcfhij+abcfgil+abcfghl+abcfghi+abcejkl+abceijk+abcehjk+abcehik+abcegjl+abcegij+abceghk+abceghj+abceghi+abcefkl+abcefhk+abcefgl+abcefgh+abcdhij+abcdghi+abcdfil+abcdfhl+abcdfhj+abcdfgh+abcdejk+abcdehk+abcdegk+abcdegj+abcdegi+abcdijkl+abcdhijl+abcdhijk+abcdgijl+abcdghjl+abcdghik+abcdfikl+abcdfijl+abcdfhil+abcdfhik+abcdfhij+abcdfgil+abcdfghl+abcdfghi+abcdejkl+abcdeijk+abcdehjk+abcdehik+abcdegjl+abcdegij+abcdeghk+abcdeghj+abcdeghi+abcdefkl+abcdefhk+abcdefgl+abcdefgh$
\item $jl+ik+dl+dk+dj+di+dh+dg+df+de+cl+ck+cj+ci+ch+cg+cf+ce+bl+bk+bj+bi+bh+bg+bf+be+al+ak+aj+ai+ah+ag+af+ae+jkl+ikl+ijl+ijk+hjl+hik+gjl+gik+fjl+fik+ejl+eik+dkl+djk+dil+dik+dij+dhl+dhk+dhj+dhi+dgl+dgk+dgj+dgi+dgh+dfl+dfk+dfj+dfi+dfh+dfg+del+dek+dej+dei+deh+deg+def+ckl+cjl+cjk+cil+cij+chl+chk+chj+chi+cgl+cgk+cgj+cgi+cgh+cfl+cfk+cfj+cfi+cfh+cfg+cel+cek+cej+cei+ceh+ceg+cef+bkl+bjk+bil+bik+bij+bhl+bhk+bhj+bhi+bgl+bgk+bgj+bgi+bgh+bfl+bfk+bfj+bfi+bfh+bfg+bel+bek+bej+bei+beh+beg+bef+akl+ajl+ajk+ail+aij+ahl+ahk+ahj+ahi+agl+agk+agj+agi+agh+afl+afk+afj+afi+afh+afg+ael+aek+aej+aei+aeh+aeg+aef+hjkl+hikl+hijl+hijk+gjkl+gikl+gijl+gijk+ghjl+ghik+fjkl+fikl+fijl+fijk+fhjl+fhik+fgjl+fgik+ejkl+eikl+eijl+eijk+ehjl+ehik+egjl+egik+efjl+efik+dikl+dijk+dhkl+dhjk+dhil+dhik+dhij+dgkl+dgjk+dgil+dgik+dgij+dghl+dghk+dghj+dghi+dfkl+dfjk+dfil+dfik+dfij+dfhl+dfhk+dfhj+dfhi+dfgl+dfgk+dfgj+dfgi+dfgh+dekl+dejk+deil+deik+deij+dehl+dehk+dehj+dehi+degl+degk+degj+degi+degh+defl+defk+defj+defi+defh+defg+cjkl+cijl+chkl+chjl+chjk+chil+chij+cgkl+cgjl+cgjk+cgil+cgij+cghl+cghk+cghj+cghi+cfkl+cfjl+cfjk+cfil+cfij+cfhl+cfhk+cfhj+cfhi+cfgl+cfgk+cfgj+cfgi+cfgh+cekl+cejl+cejk+ceil+ceij+cehl+cehk+cehj+cehi+cegl+cegk+cegj+cegi+cegh+cefl+cefk+cefj+cefi+cefh+cefg+bikl+bijk+bhkl+bhjk+bhil+bhik+bhij+bgkl+bgjk+bgil+bgik+bgij+bghl+bghk+bghj+bghi+bfkl+bfjk+bfil+bfik+bfij+bfhl+bfhk+bfhj+bfhi+bfgl+bfgk+bfgj+bfgi+bfgh+bekl+bejk+beil+beik+beij+behl+behk+behj+behi+begl+begk+begj+begi+begh+befl+befk+befj+befi+befh+befg+bdjl+ajkl+aijl+ahkl+ahjl+ahjk+ahil+ahij+agkl+agjl+agjk+agil+agij+aghl+aghk+aghj+aghi+afkl+afjl+afjk+afil+afij+afhl+afhk+afhj+afhi+afgl+afgk+afgj+afgi+afgh+aekl+aejl+aejk+aeil+aeij+aehl+aehk+aehj+aehi+aegl+aegk+aegj+aegi+aegh+aefl+aefk+aefj+aefi+aefh+aefg+acik+ghjkl+ghikl+ghijl+ghijk+fhjkl+fhikl+fhijl+fhijk+fgjkl+fgikl+fgijl+fgijk+fghjl+fghik+ehjkl+ehikl+ehijl+ehijk+egjkl+egikl+egijl+egijk+eghjl+eghik+efjkl+efikl+efijl+efijk+efhjl+efhik+efgjl+efgik+dhikl+dhijk+dgikl+dgijk+dghkl+dghjk+dghil+dghik+dghij+dfikl+dfijk+dfhkl+dfhjk+dfhil+dfhik+dfhij+dfgkl+dfgjk+dfgil+dfgik+dfgij+dfghl+dfghk+dfghj+dfghi+deikl+deijk+dehkl+dehjk+dehil+dehik+dehij+degkl+degjk+degil+degik+degij+deghl+deghk+deghj+deghi+defkl+defjk+defil+defik+defij+defhl+defhk+defhj+defhi+defgl+defgk+defgj+defgi+defgh+chjkl+chijl+cgjkl+cgijl+cghkl+cghjl+cghjk+cghil+cghij+cfjkl+cfijl+cfhkl+cfhjl+cfhjk+cfhil+cfhij+cfgkl+cfgjl+cfgjk+cfgil+cfgij+cfghl+cfghk+cfghj+cfghi+cejkl+ceijl+cehkl+cehjl+cehjk+cehil+cehij+cegkl+cegjl+cegjk+cegil+cegij+ceghl+ceghk+ceghj+ceghi+cefkl+cefjl+cefjk+cefil+cefij+cefhl+cefhk+cefhj+cefhi+cefgl+cefgk+cefgj+cefgi+cefgh+bhikl+bhijk+bgikl+bgijk+bghkl+bghjk+bghil+bghik+bghij+bfikl+bfijk+bfhkl+bfhjk+bfhil+bfhik+bfhij+bfgkl+bfgjk+bfgil+bfgik+bfgij+bfghl+bfghk+bfghj+bfghi+beikl+beijk+behkl+behjk+behil+behik+behij+begkl+begjk+begil+begik+begij+beghl+beghk+beghj+beghi+befkl+befjk+befil+befik+befij+befhl+befhk+befhj+befhi+befgl+befgk+befgj+befgi+befgh+bdjkl+bdijl+bdhjl+bdgjl+bdfjl+bdejl+ahjkl+ahijl+agjkl+agijl+aghkl+aghjl+aghjk+aghil+aghij+afjkl+afijl+afhkl+afhjl+afhjk+afhil+afhij+afgkl+afgjl+afgjk+afgil+afgij+afghl+afghk+afghj+afghi+aejkl+aeijl+aehkl+aehjl+aehjk+aehil+aehij+aegkl+aegjl+aegjk+aegil+aegij+aeghl+aeghk+aeghj+aeghi+aefkl+aefjl+aefjk+aefil+aefij+aefhl+aefhk+aefhj+aefhi+aefgl+aefgk+aefgj+aefgi+aefgh+acikl+acijk+achik+acgik+acfik+aceik+fghjkl+fghikl+fghijl+fghijk+eghjkl+eghikl+eghijl+eghijk+efhjkl+efhikl+efhijl+efhijk+efgjkl+efgikl+efgijl+efgijk+efghjl+efghik+dghikl+dghijk+dfhikl+dfhijk+dfgikl+dfgijk+dfghkl+dfghjk+dfghil+dfghik+dfghij+dehikl+dehijk+degikl+degijk+deghkl+deghjk+deghil+deghik+deghij+defikl+defijk+defhkl+defhjk+defhil+defhik+defhij+defgkl+defgjk+defgil+defgik+defgij+defghl+defghk+defghj+defghi+cghjkl+cghijl+cfhjkl+cfhijl+cfgjkl+cfgijl+cfghkl+cfghjl+cfghjk+cfghil+cfghij+cehjkl+cehijl+cegjkl+cegijl+ceghkl+ceghjl+ceghjk+ceghil+ceghij+cefjkl+cefijl+cefhkl+cefhjl+cefhjk+cefhil+cefhij+cefgkl+cefgjl+cefgjk+cefgil+cefgij+cefghl+cefghk+cefghj+cefghi+bghikl+bghijk+bfhikl+bfhijk+bfgikl+bfgijk+bfghkl+bfghjk+bfghil+bfghik+bfghij+behikl+behijk+begikl+begijk+beghkl+beghjk+beghil+beghik+beghij+befikl+befijk+befhkl+befhjk+befhil+befhik+befhij+befgkl+befgjk+befgil+befgik+befgij+befghl+befghk+befghj+befghi+bdijkl+bdhjkl+bdhijl+bdgjkl+bdgijl+bdghjl+bdfjkl+bdfijl+bdfhjl+bdfgjl+bdejkl+bdeijl+bdehjl+bdegjl+bdefjl+aghjkl+aghijl+afhjkl+afhijl+afgjkl+afgijl+afghkl+afghjl+afghjk+afghil+afghij+aehjkl+aehijl+aegjkl+aegijl+aeghkl+aeghjl+aeghjk+aeghil+aeghij+aefjkl+aefijl+aefhkl+aefhjl+aefhjk+aefhil+aefhij+aefgkl+aefgjl+aefgjk+aefgil+aefgij+aefghl+aefghk+aefghj+aefghi+acijkl+achikl+achijk+acgikl+acgijk+acghik+acfikl+acfijk+acfhik+acfgik+aceikl+aceijk+acehik+acegik+acefik+efghjkl+efghikl+efghijl+efghijk+dfghikl+dfghijk+deghikl+deghijk+defhikl+defhijk+defgikl+defgijk+defghkl+defghjk+defghil+defghik+defghij+cfghjkl+cfghijl+ceghjkl+ceghijl+cefhjkl+cefhijl+cefgjkl+cefgijl+cefghkl+cefghjl+cefghjk+cefghil+cefghij+bfghikl+bfghijk+beghikl+beghijk+befhikl+befhijk+befgikl+befgijk+befghkl+befghjk+befghil+befghik+befghij+bdhijkl+bdgijkl+bdghjkl+bdghijl+bdfijkl+bdfhjkl+bdfhijl+bdfgjkl+bdfgijl+bdfghjl+bdeijkl+bdehjkl+bdehijl+bdegjkl+bdegijl+bdeghjl+bdefjkl+bdefijl+bdefhjl+bdefgjl+afghjkl+afghijl+aeghjkl+aeghijl+aefhjkl+aefhijl+aefgjkl+aefgijl+aefghkl+aefghjl+aefghjk+aefghil+aefghij+achijkl+acgijkl+acghikl+acghijk+acfijkl+acfhikl+acfhijk+acfgikl+acfgijk+acfghik+aceijkl+acehikl+acehijk+acegikl+acegijk+aceghik+acefikl+acefijk+acefhik+acefgik+defghikl+defghijk+cefghjkl+cefghijl+befghikl+befghijk+bdghijkl+bdfhijkl+bdfgijkl+bdfghjkl+bdfghijl+bdehijkl+bdegijkl+bdeghjkl+bdeghijl+bdefijkl+bdefhjkl+bdefhijl+bdefgjkl+bdefgijl+bdefghjl+aefghjkl+aefghijl+acghijkl+acfhijkl+acfgijkl+acfghikl+acfghijk+acehijkl+acegijkl+aceghikl+aceghijk+acefijkl+acefhikl+acefhijk+acefgikl+acefgijk+acefghik+bdfghijkl+bdeghijkl+bdefhijkl+bdefgijkl+bdefghjkl+bdefghijl+acfghijkl+aceghijkl+acefhijkl+acefgijkl+acefghikl+acefghijk+bdefghijkl+acefghijkl$
\end{enumerate}
\end{footnotesize}
\end{spacing}

\chapter{Github Project weak\_keys}
The scripts used for this thesis can be found in \url{https://github.com/mariosgeorgiou17uclacuk/weak_keys}
\section{Walsh Spectrum and Autocorrelation Spectrum}
We use Python3 to compute the Walsh Spectrum and Autocorrelation Spectrum. By executing the commands below, we can compute the Walsh and Autocorrelation Spectrum for the Boolean function of T-310:
\begin{lstlisting}
pip install itertools
pip install numpy
pip install scipy
\end{lstlisting}
We then run the python script by:
\begin{lstlisting}
python spectrum.py
\end{lstlisting}
The results will be printed on the screen.
\section{Invariants on 12 bits}
In order to compute all invariants on 12 bits we need first to create a text file, we name it 'IOquestion12.txt', of this format:
\begin{figure}[h]
    \centering
    \includegraphics[scale=0.7]{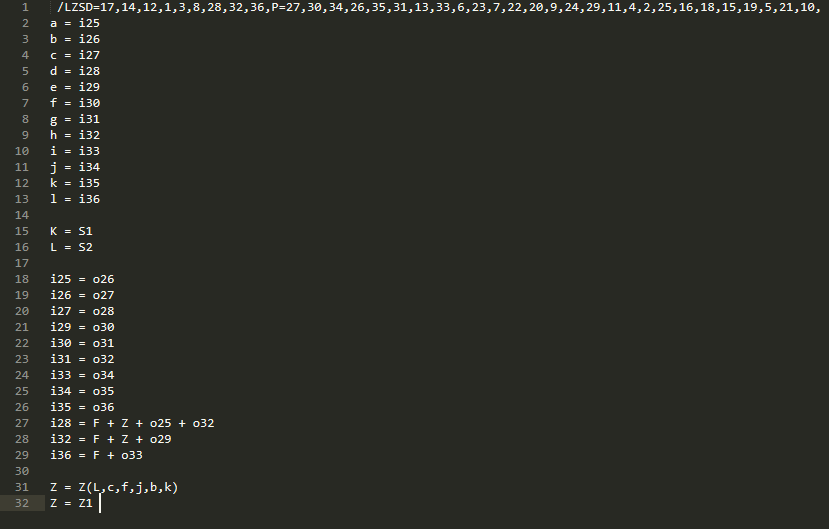}
    \caption{IOquestion12.txt}
    \label{fig:mesh1}
\end{figure}
\newline
We use Python3 and Sagemath 8.0, so before running the files we need to execute the commands below:
\begin{lstlisting}
pip install itertools
pip install sys
pip install compiler
\end{lstlisting}
\vspace{-3pt}
Step 1: We run mongen.py by:
\begin{lstlisting}
python mongen.py IOquestion12.txt
\end{lstlisting}
\vspace{-3pt}
which takes as input the text file 'IOquestion12.txt', and it will create two separate text files, called 'draft1.txt' and 'draft2.txt'. 'draft1.txt' contains all possible monomials for the variables $\{a,b,c,\dots,k,l\}$ and 'draft2.txt' contains all the resulting polynomials after one round, for each possible monomial in 'draft1.txt'.\\
Step 2: We then run mongen.ipynb, a Sagemath file, which calculates all resulting polynomials by removing the parentheses from 'draft2.txt'. It will then create a text file, called\\ 'IOquestion12.all\_monomials.txt', which contains two columns of data. The first column contains all the possible monomials from 'draft1.txt' and the second column contains the corresponding polynomials after one round.\\\\
Step 3: When the new text file is created, we run ax64.exe as follows:
\begin{small}
\begin{lstlisting}
ax64.exe 41012 "IOmonomials.temp.txt" "IOquestion12.all_monomials.txt"
\end{lstlisting}
\end{small}
to create a text file, called "IOmonomials.temp.txt", which contains the XOR of the two columns.\\\\
Step 4: As soon as "IOmonomials.temp.txt" is created, we run replacebooleanfunction.ipynb, a Sagemath file, which will substitute F,L and Z, do the calculations for each polynomial and write the result in the text file called 'IOmonomials.temp2.txt'.\\\\
Step 5: Finally, we use ax64.exe again by executing the command below:
\begin{small}
\begin{lstlisting}
ax64.exe 41013 "IOmonomials.temp.rewritten.txt" "IOmonomials.temp2.txt" \end{lstlisting}
\end{small}
A text file, called 'Kernel\_abcd.txt' will be created, which will contain all the invariants for the specific long term key.

\end{document}